\documentclass[11pt]{article}

\usepackage[margin=1in]{geometry}
\usepackage[utf8]{inputenc}
\usepackage{microtype}
\usepackage{amsmath,amssymb,amsthm,mathtools,bbm}
\usepackage{xcolor}
\usepackage[colorlinks=true,linkcolor=blue,citecolor=blue,urlcolor=blue]{hyperref}
\usepackage{algorithm}
\PassOptionsToPackage{noend}{algpseudocode}
\usepackage{algpseudocode}
\usepackage{booktabs}
\usepackage{enumitem}

\algrenewcommand\algorithmicrequire{\textbf{Input:}}
\algrenewcommand\algorithmicensure{\textbf{Output:}}

\newtheorem{theorem}{Theorem}[section]
\newtheorem{lemma}[theorem]{Lemma}

\newtheorem{corollary}[theorem]{Corollary}

\newtheorem{remark}[theorem]{Remark}

\newtheorem*{theorem*}{Theorem}

\usepackage[nameinlink,noabbrev]{cleveref}

\crefname{theorem}{Theorem}{Theorems}
\Crefname{theorem}{Theorem}{Theorems}
\crefname{lemma}{Lemma}{Lemmas}
\Crefname{lemma}{Lemma}{Lemmas}
\crefname{proposition}{Proposition}{Propositions}
\Crefname{proposition}{Proposition}{Propositions}
\crefname{corollary}{Corollary}{Corollaries}
\Crefname{corollary}{Corollary}{Corollaries}
\crefname{claim}{Claim}{Claims}
\Crefname{claim}{Claim}{Claims}
\crefname{definition}{Definition}{Definitions}
\Crefname{definition}{Definition}{Definitions}
\crefname{remark}{Remark}{Remarks}
\Crefname{remark}{Remark}{Remarks}
\crefname{conjecture}{Conjecture}{Conjectures}
\Crefname{conjecture}{Conjecture}{Conjectures}
\crefname{openproblem}{Open Problem}{Open Problems}
\Crefname{openproblem}{Open Problem}{Open Problems}
\crefname{algorithm}{Algorithm}{Algorithms}
\Crefname{algorithm}{Algorithm}{Algorithms}

\newenvironment{question}{%
  \nopagebreak\vspace{1em}
  \begin{quote}
  \itshape
  \noindent\textbf{}
}{%
  \end{quote}%
  \vspace{1.5em}
}

\newcommand{\E}{\mathop{\mathbb E}}
\newcommand{\Prb}{\mathop{\mathbb P}}
\newcommand{\Var}{\mathop{\mathrm{Var}}}
\newcommand{\eps}{\varepsilon}
\newcommand{\Z}{\mathbb Z}
\newcommand{\Zn}{\mathbb Z_n}

\newcommand{\Bin}{\operatorname{Bin}}

\newcommand{\one}{\mathbbm{1}}

\newcommand{\OPT}{\operatorname{OT}^\star}

\title{Locality in Open Addressing Hash Tables}
\author{Or Zamir \\ Tel Aviv University}
\date{}

\begin{document}
\maketitle

\begin{abstract}
Open-addressed hash tables without reordering, such as linear probing and
uniform probing, are among the simplest and most widely used data structures.
Their performance is traditionally measured by probe count.  We propose
studying a complementary parameter: \emph{locality}, defined as the geometric
distance from the first probed location to the farthest cell inspected or used.

At load factor $1-\varepsilon$, uniform probing achieves the optimal
$\Theta(1/\varepsilon)$ probe count among greedy schemes, but has essentially
no locality.  Linear probing is highly local, but performs
$\Theta(1/\varepsilon^2)$ probes.  Our main result shows that this quadratic
locality scale is fundamental: for open addressing without reordering, no
algorithm can achieve locality $o(1/\varepsilon^2)$ simultaneously at all
loads $1-\varepsilon$.  We further prove an amortized expected-locality lower
bound of $\Omega(1/\varepsilon)$ over any sequence of
$(1-\varepsilon)n$ insertions, even when the final load is known in advance.

As locality captures cache and external-memory behavior, our result helps
explain the practical popularity of linear probing.  It also implies that page
size $B=\Omega(1/\varepsilon^2)$ is necessary for $1+o(1)$ expected page span
in immutable open addressing.

We complement these lower bounds with two upper bounds.  When the target load
is known in advance, we show that the amortized lower bound can essentially be
deamortized: every insertion and every successful or unsuccessful search has
expected probe count and expected locality
$
        \widetilde O(1/\varepsilon).
$
That is, the costs of the final insertions can be redistributed across the
entire insertion sequence.

We also give a load-oblivious greedy scheme with optimal expected probe count
$\Theta(1/\varepsilon)$, in which the $i$-th probe is at distance $O(i^2)$
from the first probe.  Its analysis relies on a general variance bound for
occupied-cell densities in symmetric probing schemes, which also implies an
$O(\log n/\varepsilon^2)$ expected probe bound for every fixed-shift probing
sequence and every load $1-\varepsilon$.

\end{abstract}

\newpage
\tableofcontents
\newpage

\section{Introduction}

Hash tables are among the most basic and widely used data structures.
In the open-addressing paradigm, all keys are stored directly in a single array,
and collisions are resolved by probing additional cells of the same table.
This makes open addressing simple, space-efficient, and particularly attractive
in settings where pointer chasing, memory allocation, and cache misses dominate
the running time~\cite{Knuth1998,CormenLeisersonRivestStein2009,PurcellHarris2005,
AskitisZobel2005}.  
Classical examples include linear probing -- which is still extremely dominant in practice, as well as uniform probing, double hashing, and quadratic probing.

The traditional theoretical measure for such hash tables is the number of
probes made by a query.  With this measure, the
high-load behavior of different open-addressing schemes is very different.  At
load \(1-\eps\), uniform probing has expected probe count \(\Theta(1/\eps)\),
whereas linear probing has expected insertion time \(\Theta(1/\eps^2)\) due to
primary clustering~\cite{Knuth1998,FlajoletPobleteViola1998,
BenderKuszmaulKuszmaul2021}.  Much of the modern theory of open addressing is
concerned with understanding, and improving, such probe-count bounds.

Recent years have seen striking progress in this direction.  Bender, Kuszmaul,
and Zhou introduced rainbow hashing and obtained tight bounds for classical
open addressing, including \(O(1)\) expected-time queries and
\(O(\log\log \eps^{-1})\) expected-time updates at load \(1-\eps\)
\cite{BenderKuszmaulZhou2024}.  Farach-Colton, Krapivin, and Kuszmaul
subsequently showed that even without reordering elements over time -- a property highly preferred in practice for various reasons, one can
achieve probe complexities far better than previously believed, by using non-greedy schemes that thus circumvent Yao's seminal lower bound in his work on uniform
hashing~\cite{Yao1985,FarachColtonKrapivinKuszmaul2025}.  A further recent
line studies the power and limitations of greedy open-addressing schemes
themselves~\cite{FarachColtonKrapivinKuszmaul2026}.  These results make clear
that, if the only objective is probe count, then linear probing is far from
optimal.

This comes in contrast to that prevailing dominance of linear probing in practice, which stems from probe count not being the only relevant measure.  
Linear probing seems to consistently outperform alternatives because it is extremely local: once the first probe for a key is computed, the algorithm scans a contiguous region
of memory.  
This locality interacts well with cache lines, memory bandwidth, prefetching, and external-memory page layout.  
Uniform probing has the opposite behavior: it obtains the \(\Theta(1/\eps)\) probe count, which is optimal for purely greedy algorithms, but the probes
are essentially spread throughout the whole table.
Other schemes that optimize probe counts can likewise lose geometric locality.
For example, bucketized cuckoo hashing inspects only two contiguous buckets per
lookup, but these buckets may be far apart, and insertions may relocate
previously stored keys~\cite{DietzfelbingerWeidling2007}.
Thus it is natural to ask whether the quadratic behavior of linear probing is merely a weakness of linear probing, or whether it reflects a more fundamental locality barrier.

In this paper we initiate the systematic study of this second parameter, in open addressing tables without reordering.  
We define the \emph{locality} of an operation to be the geometric distance from the
first probed location to the farthest cell inspected, including the cell ultimately
used to store the key.  The main message of the paper is that while the probe count can be
dramatically improved, \emph{the locality of linear probing is optimal}: for open addressing without reordering, expected locality of \(\Omega(1/\eps^2)\) is necessary in the standard setting in which the maximum load is not known in advance.

\subsection{Our Results}

We study open-addressing tables in which elements are never moved after they are
inserted.  This is the no-reordering or immutable setting which has numerous practical advantages and is thus very popular.
The insertion algorithm may be adaptive, randomized, and depend on what it sees in table cells; it need not be greedy and it need not follow a fixed probe sequence.  The only structural requirement is that once an inserted key is stored in a cell, it remains there.

Our main theorem says that no such algorithm can beat the quadratic locality scale simultaneously at all loads.

\begin{theorem*}[Main locality lower bound, informal]
No open-addressing algorithm without reordering can have expected insertion
locality \(o(1/\eps^2)\) at every load \(1-\eps\).  
\end{theorem*}

The all-load quantifier is natural when insertion bounds are used to control
search costs: an item found at load $1-\eps$ may have been inserted earlier,
at a different load, so the relevant insertion guarantee is a profile that
holds throughout the evolution of the table.

We also prove a lower bound that applies even when the final load is fixed in
advance.  Namely, the average expected locality over all insertions up to load
$1-\eps$ must be $\Omega(1/\eps)$.

\begin{theorem*}[Amortized locality lower bound, informal]
Any open-addressing algorithm without reordering must have amortized expected
insertion locality $\Omega(1/\eps)$ across a sequence of
$(1-\eps)n$ insertions, even when the algorithm is given $\eps$ in advance.
\end{theorem*}

Both of the above bounds are matched by linear probing.
The same locality lower bound has an external-memory consequence.  In the
external-memory model, memory is divided into pages of \(B\) consecutive cells.
The classical high-load threshold associated with \(1+o(1)\) page accesses is
\(B=\Theta(1/\eps^2)\): Jensen and Pagh achieved
\(1+O(1/\sqrt B)\) expected I/Os and space overhead, and Bender, Kuszmaul, and
Kuszmaul identify \(B=\Omega(x^2)\), for \(x=1/\eps\), as the long-standing barrier
for \(1+o(1)\) block transfers at load \(1-1/x\)
\cite{AggarwalVitter1988,JensenPagh2008,BenderKuszmaulKuszmaul2021}.  In the
immutable no-reordering setting, our lower bound gives a geometric explanation
for this quadratic threshold.

\begin{theorem*}[External-memory locality lower bound, informal]
For immutable open addressing at load \(1-\eps\), page size
\[
        B=\Omega(1/\eps^2)
\]
is necessary for \(1+o(1)\) expected page span.
\end{theorem*}

This should be contrasted with schemes that maintain additional structure by
using rebuilding or reordering.  Graveyard hashing uses
frequent rebuilding to create an anti-clustering effect and obtains \(1+o(1)\) expected
block transfers when \(x=o(B)\), beating the previous \(x^2\) threshold
\cite{BenderKuszmaulKuszmaul2021}.  Zombie hashing later deamortizes this approach by redistributing intentionally-empty cells incrementally
\cite{ChesettiShiPhillipsPandey2025}.  
These results are complementary to ours: our theorem shows that without reordering or such
maintenance, the quadratic locality scale is inherent.

We complement these lower bounds with two upper bounds.  We first return to
the setting in which the target load $1-\eps$ is known in advance.  The
amortized lower bound leaves open whether the cost of the final insertions can
be shifted to earlier ones, so that every individual operation has expected
cost close to $1/\eps$.  We show that this can indeed be done, up to
logarithmic factors.

\begin{theorem*}[Known-load upper bound, informal]
If the target slack $\eps$ is known in advance, there is an immutable
open-addressing scheme such that every insertion before load $1-\eps$, and
every successful or unsuccessful search at load at most $1-\eps$, satisfies
\[
        \E[T],\ \E[R]
        =
        O\left(\frac{\log^3(1/\eps)}{\eps}\right).
\]
\end{theorem*}

Thus the amortized $\Omega(1/\eps)$ locality lower bound can essentially be
deamortized across the insertion sequence.

We next ask what can be achieved without knowing the target load in advance.
While linear probing matches the all-load locality lower bound, its probe count
is suboptimal.  Can one simultaneously obtain the optimal expected probe count
and retain a local probe sequence?

\begin{theorem*}[Load-oblivious greedy upper bound, informal]
There is a greedy probing scheme such that, at any load \(1-\eps\), the expected number of probes is
\[
        \E[T]=O(1/\eps).
\]
Moreover, the \(i\)-th probe is always at distance \(O(i^2)\) from the home
location, and hence the locality radius of an operation satisfies \(R=O(T^2)\)
deterministically.
\end{theorem*}

This theorem gives unconditionally the behavior that one would obtain from
quadratic probing if its conjectured high-load probe count holds.
Quadratic probing was introduced by Maurer~\cite{Maurer1968} and has been
studied and refined in several classical works~\cite{Radke1970,
HopgoodDavenport1972,Ecker1974,Batagelj1975}.  Despite its simplicity and
practical appeal, its rigorous analysis remains elusive: Kuszmaul and Xi
recently proved the first nontrivial positive result, showing constant expected
time at some positive constant load factor~\cite{KuszmaulXi2024}.  The
high-load behavior remains open.

Our proof of the load-oblivious upper bound rests on a general variance lemma
for symmetric probing schemes.  Although the table state is produced by a
highly adaptive insertion process, we show that for every fixed set
\(S\subseteq\Zn\), if we denote by~$O_t$ the set of occupied table cells then
\[
        \Var(|O_t\cap S|)\le |S|.
\]
This is a surprisingly robust pseudorandomness property of the occupied set:
although the insertions are highly adaptive and the final table may be far from
a set of independent Bernoulli cells, the number of occupied cells in any fixed
test set has variance no larger than the size of the set.
As a byproduct, we get a weak all-load guarantee for fixed-offset
probing sequences: 
\begin{theorem*}[Weak upper bound for any symmetric probe sequence, informal]
For any fixed sequence of shifts, which may also be dependent random variables, at load
\(1-\eps\), the expected number of probes is
\[
        \E \left[T\right]=O(\eps^{-2}\log n).
\]
\end{theorem*}
This is far weaker than the constant-load theorem of Kuszmaul and
Xi~\cite{KuszmaulXi2024}, but it applies at every load and also for randomized shift sequences.
In particular, if one seeks an \(n\)-independent worst-case bound
\(f(1/\varepsilon)\) for any arbitrary fixed-offset probing, then in every
regime \(1/\varepsilon=(\log n)^K\) our result already gives
\[
    f(x) \le O\!\left(x^{2+1/K}\right),
    \qquad x=1/\varepsilon .
\]
Thus any genuinely worse polynomial or superpolynomial behavior cannot
persist uniformly across all loads.

The load-oblivious upper bound still lacks a desired property.  Although
$\E[T]=O(1/\eps)$ and $R=O(T^2)$ imply locality $O(1/\eps^2)$ with arbitrarily
high constant probability, they do not give the same bound in expectation:
a first-moment bound on $T$ is insufficient to control $\E[T^2]$.

The known-load theorem stated above obtains an expected-locality guarantee by
redistributing the cost across the insertion sequence.  Whether one can obtain
a comparable expected-locality guarantee without knowing the target load
remains open.

\subsection{Overview and Organization}

Both our lower and upper bounds stem from simple reductions to natural structural statements on open-addressed hash tables which we introduce and prove. These statements may be of independent interest.


\paragraph{Lower bounds:}
The technical core of the lower bound is a physics-inspired transport statement.  Suppose that
\(m=(1-\eps)n\) home locations are sampled independently from an arbitrary
distribution~$\mathcal D$ on the cycle \(\Zn\), which is paired with the natural metric on it. A ``particle'' is placed on each of these home locations. 
What is the minimum total movement needed to move
all particles to distinct locations in $\Zn$?  We prove that the answer is always
\(\Omega(n/\eps)\).

\begin{theorem*}[Transport lower bound, informal]
	Let \(m=(1-\eps)n\), and assume \(n\eps^2\) is sufficiently large.  For every
	distribution \(\mathcal D\) on \(\Zn\), if
	\(H_1,\ldots,H_m\sim\mathcal D\) independently, then the expected optimal
	transport cost from \(H_1,\ldots,H_m\) to distinct table cells is
	\[
	\Omega(n/\eps).
	\]
\end{theorem*}

This theorem can be viewed as a one-dimensional and discrete hard-core transport statement:
randomly dropped particles must be moved to a configuration in which no cell is used twice, corresponding to no two particles being too close.  
Several questions of similar nature were extensively studied in mathematical physics, albeit usually with other metric spaces (e.g., continuous rather than discrete, or three-dimensional rather than one-dimensional) and other cost measures (e.g., time until we converge at a steady state in which every pair of particles is sufficiently far, or the structure of such stable states).

Section~\ref{sec:transport} proves the transport lower bound. 
We start by analyzing this optimal transport cost when~$\mathcal D$ is the uniform distribution over~$\Zn$. The proof identifies the unavoidable cost of local density fluctuations: intervals of length~$\Theta(1/\varepsilon^2)$ often have an overflowing number of home locations mapped into them. Then, using convexity arguments we show that any other distribution~$\mathcal D$ incurs a strictly larger cost than the uniform distribution.  

Section~\ref{sec:hashing-lb} converts transport into hashing lower bounds.  
The hashing lower bound follows by observing that every
immutable open-addressing algorithm produces some distinct final placement, and the locality radius of each insertion is dominated by the final displacement of the inserted key.
We insert independent uniformly random elements from a universe of size \(\omega(n^2)\).  With high probability the inserted elements are distinct, and
conditioned on the table initialization their first probe locations are i.i.d.
from some distribution on \(\Zn\).  The distribution-free transport theorem
therefore applies.  This immediately gives the amortized \(\Omega(n/\eps)\) total locality
lower bound.  The all-load \(1/\eps^2\) lower bound follows by an integration argument.

\paragraph{Upper bounds:}
The analysis of local or fixed-shifts probe sequences is notoriously hard. As mentioned before, it was only recently that any constant expectation upper bound for any constant load factor was proven for any fixed sequence of shifts that is not equivalent to linear probing~\cite{KuszmaulXi2024}.
The technical core of our upper bound is a general statement on the behavior of a large family of probing algorithms.
Consider any open-addressed probing algorithm that is both greedy and \emph{symmetric} or \emph{translation invariant}: that is, the home location is uniform and the sequence of shifts between subsequent probes and the home location is drawn independently of the home location itself. 

Denote by \(O_t\) the set of occupied cells after \(t\) insertions, and consider an arbitrary fixed set \(I\subseteq\Zn\).
Due to symmetry, it is clear that the expected number of occupied cells within~$I$ is proportional to the load of the entire table
\[
\E\left[|O_t\cap I|\right] = \sum_{i\in I} \Pr\left[i \in O_t\right] = \frac{t}{n}\cdot |I|
.
\]
Our main technical point is bounding the variance of this quantity as well.

\begin{theorem*}[Variance upper bound for symmetric schemes, informal]
	In any symmetric probing scheme, if
	\(O_t\) is the occupied set after \(t\) insertions, then for every fixed set
	\(I\subseteq\Zn\),
	\[
	\Var(|O_t\cap I|)\le |I|.
	\] 
\end{theorem*}
Note that when~$t=\Omega(n)$ this lemma implies that the number of occupied cells within any set~$I$ behaves in both expectation and variance similarly to a binomial distribution~$\text{Bin}\left(|I|,\frac{t}{n}\right)$.

In Section~\ref{sec:upper} we begin by proving the above lemma. We also derive as a simple corollary the~$O(\log n / \varepsilon^2)$ upper bound on the expected probe count in any symmetric scheme.  
Furthermore, it implies that in any symmetric scheme, a random interval of length
\(\Theta(1/\eps^2)\) contains an \(\Omega(\eps)\) fraction of empty cells with
constant probability.  
This is enough for us to introduce and analyze a natural load-oblivious ``expanding-window sampler'' to find
an empty cell in expected \(O(1/\eps)\) probes: 
It enumerates guesses
\(\eps=1/2,1/4,1/8,\ldots\), and for each guess samples
\(\Theta(1/\eps)\) cells uniformly from an interval of length
\(\Theta(1/\eps^2)\) around the home location.  This achieves the optimal
expected probe count, and its probe sequence has the same quadratic locality
profile as quadratic probing.
 
The same argument does not prove that the expected locality is
$\E[R]=O(1/\eps^2)$, because $R=O(T^2)$ and a first-moment bound on $T$ does
not imply a second-moment bound.  Obtaining the matching expected locality for
this load-oblivious greedy sampler would require stronger tail bounds on the
density of occupied cells.

Section~\ref{sec:known-upper} considers the complementary setting in which the
final slack $\eps$ is known in advance.  The main idea is to replace the
discrepancy in cost over time by a probabilistic discrepancy for each
insertion.  Most keys are kept close to their home locations, while a small
fraction intentionally pay a larger cost and are sent to a least-loaded larger
region.  These more expensive insertions leave enough local space for the
remaining keys.

We first implement this idea using two block sizes, obtaining expected
insertion and search cost
\[
        O\left(\frac{\log(1/\eps)}{\eps^{3/2}}\right).
\]
We then repeat the same cost-redistribution step over multiple scales.  A
biased minimum-allocation lemma controls the imbalance within each block,
while a block-level linear-probing estimate controls overflow between the
largest blocks.  The resulting construction satisfies
\[
        \E[T],\ \E[R]
        =
        O\left(\frac{\log^3(1/\eps)}{\eps}\right)
\]
for every insertion and every successful or unsuccessful search.

Section~\ref{sec:model} formally defines our model and cost measures.
Sections~\ref{sec:transport} and~\ref{sec:hashing-lb} prove the transport and
hashing lower bounds.  Section~\ref{sec:upper} proves the variance lemma and
analyzes the load-oblivious greedy upper bound, while
Section~\ref{sec:known-upper} gives the known-load construction.  Finally,
Section~\ref{sec:open-problems} discusses the remaining gaps, including whether
the logarithmic factors in the known-load upper bound can be removed and
whether the load-oblivious scheme admits an expected-locality guarantee.

\section{Model, Notation, and Preliminaries}\label{sec:model}

We identify the table with the cyclic group
\[
        \Zn=\{0,1,\ldots,n-1\},
\]
and all table indices are understood modulo \(n\).  For \(a,b\in\Zn\), let
\[
        |a-b|_n=\min_{k\in\Z}|a-b+kn|
\]
be their cyclic distance.  For an integer \(L\le n\) and a start point
\(s\in\Zn\), we write
\[
        I_s^{(L)}=\{s,s+1,\ldots,s+L-1\}\subseteq\Zn
\]
for the cyclic interval of length \(L\) starting at \(s\).  When \(L\) is clear
from context we simply write \(I_s\).  For a real number \(x\), let
\[
        (x)_+=\max\{x,0\}.
\]

We write \(\E\), \(\Prb\), and \(\Var\) for expectation, probability, and
variance.  We write \(U(S)\) for the uniform distribution on a finite set
\(S\), and \(\Bin(m,p)\) for a binomial random variable with \(m\) trials and
success probability \(p\).  Unless stated otherwise, all constants hidden in
\(O(\cdot)\), \(\Omega(\cdot)\), and \(\Theta(\cdot)\) are absolute.  Our lower
bounds are stated in the high-load regime
\[
        \eps\to 0,
        \qquad
        n\eps^2\to\infty .
\]
When \(1/\eps^2\) is comparable to or larger than \(n\), the locality scale
saturates at the diameter of the table, and the interesting asymptotic regime is
therefore \(n\eps^2\gg1\).

\subsection{Open addressing without reordering}

A table has \(n\) cells and stores elements from a universe \(\mathcal U_n\).
The table may perform arbitrary randomized initialization, including choosing
hash functions or any other auxiliary data.  The insertion algorithm may be
adaptive, randomized, and table-aware: while inserting a key it may inspect
cells based on the contents of the previous table cells it encountered, the auxiliary initialization data, and on its own randomness.  The only
structural restriction is \emph{no reordering}: once a key is placed in a cell,
it is never moved by later operations.

For an insertion sequence, let the \emph{home location} \(H_i\in\Zn\) be the first probed cell of the
\(i\)-th key, and let the \emph{placement} \(Y_i\in\Zn\) be the cell in which this key is finally
stored. 
The \emph{locality radius} of insertion \(i\),
denoted \(R_i\), is the maximum cyclic distance from \(H_i\) to any cell
inspected during the insertion, including the cell ultimately used.  We also
write
\[
        D_i=|H_i-Y_i|_n
\]
for the final displacement of the key.  
We usually denote by~$R$ the locality radius of an insertion and by~$T$ the probe count in it.

\subsection{Greedy and symmetric probing schemes}

For the upper bounds we consider standard greedy probing schemes.  A greedy
scheme assigns to each key a sequence of shifts
\[
        S_0,S_1,S_2,\ldots\in\Zn,
        \qquad S_0=0,
\]
and if the home location is \(H\), the probed cells are
\[
        H+S_0, H+S_1, H+S_2,\ldots .
\]
The key is stored in the first empty cell of this sequence.  The shifts may be
random and may have arbitrary dependence within one key, but the sequence is
sampled independently for different keys and before the insertion begins.  The
probe count \(T\) is the number of inspected cells, and the locality radius is
\[
        R=\max_{0\le j<T}|S_j|_n .
\]

A greedy probing scheme is \emph{symmetric}, or translation invariant, if the
joint distribution of the shift sequence \((S_0,S_1,S_2,\ldots)\) is fixed in
advance and does not depend on the home location except through the additive
shift by \(H\).  Equivalently, shifting all home locations and all occupied cells
by the same amount shifts the distribution of the entire execution by the same
amount.  This class contains linear probing, double hashing with a random step,
uniform probing, quadratic probing when it is well-defined on the table size,
and the expanding-window sampler analyzed in Section~\ref{sec:upper}.

\subsection{Probability tools}

We record the standard concentration inequalities used in the paper.  First, we
use Chebyshev's inequality: if \(X\) has finite variance, then for every
\(t>0\),
\[
        \Prb[|X-\E \left[X\right]|\ge t]
        \le
        \frac{\Var(X)}{t^2}.
\]

We use the following Chernoff bound for sums of independent Bernoulli random
variables.  If \(X\) is such a sum with mean \(\mu\), then for every \(t>0\),
\[
        \Prb[X\ge \mu+t]
        \le
        \exp\left(-\frac{t^2}{2(\mu+t/3)}\right).
\]
In particular, if \(t\le \mu\), this is at most \(\exp(-t^2/(3\mu))\).

We also use the Efron-Stein inequality.  Let
\(X=f(Z_1,\ldots,Z_m)\), where \(Z_1,\ldots,Z_m\) are independent, and let
\(X^{(i)}\) be obtained by replacing \(Z_i\) by an independent copy while
leaving all other coordinates unchanged.  Then
\[
        \Var(X)
        \le
        \frac12\sum_{i=1}^m \E\left[(X-X^{(i)})^2\right].
\]

Finally, in Lemma~\ref{lem:binomial-overload} we use only the following weak
consequence of the Berry-Esseen theorem for binomial random variables: if
\(Z\sim\Bin(m,p)\), \(\mu=\E \left[Z\right]\), and
\(\sigma^2=\Var(Z)\to\infty\), then for every fixed constant \(a\in\mathbb R\),
\[
        \Prb[Z\ge \mu+a\sigma]
\]
converges to the corresponding standard normal tail probability.  In particular,
for any fixed \(a\), this probability is bounded below by a positive absolute
constant once \(\sigma\) is sufficiently large.

\section{Transport Lower Bounds}\label{sec:transport}

Throughout this section we pair~$\Zn$ with the natural distance on the cycle, as defined in Section~\ref{sec:model}.
Let~$H,Y\sqsubseteq \Zn$ be two \emph{multi-sets} of the same cardinality~$m$, we define the \emph{optimal transport} between them as the minimum, over all perfect matchings between~$H$ and~$Y$, of the sum of distances between each element of~$H$ and its matched element in~$Y$. That is, 
\[
\text{OT}(H,Y) := 
\min_{\sigma \in S_m} \sum_{i=1}^{m} |H_i-Y_{\sigma(i)}|_n .
\]
This quantity is often also called Earth Mover's Distance or the first Wasserstein metric~$W_1$.
We call a multi-set~$Y$ \emph{distinct} if it contains no repetition (hence, it is in fact a set). 
Finally, for a multi-set~$H$ of \emph{home locations} $H_1,\ldots,H_m$, we define the optimal transport cost to distinct table cells by
\[
        \OPT(H_1,\ldots,H_m)
        =\min_{\substack{Y_1,\ldots,Y_m\in \Zn\\Y\text{ distinct}}}
          \text{OT}(H,Y).
\]
We note the relevance of this definition to our studied problem: every open-addressing table producing final locations $Y_i$ with first probe locations~$H_i$ incurs a sum of probe localities at least this optimum.

In this section, we study the following question:
\begin{question}
    Given an arbitrary distribution~$\mathcal{D}$ over~$\Zn$, what is the minimum possible~$\E\left[\OPT(H_1,\ldots,H_m)\right]$ for i.i.d. sampled~$H_1,\ldots,H_m\sim \mathcal{D}$?
\end{question}

This question is reminiscent of several themes in mathematical physics.  In hard-core particle systems, such as the one-dimensional Tonks gas of hard rods and higher-dimensional hard-sphere gases, particles interact through an exclusion constraint that forbids them from becoming too close~\cite{Tonks1936,AlderWainwright1957}.  Other repulsive-particle models, such as log gases and Coulomb gases, replace the hard constraint by a singular repulsive energy and study the resulting equilibrium, rigidity, and fluctuation behavior~\cite{Forrester2010,Serfaty2017}.  A still closer connection to the present formulation comes from density-constrained optimal-transport models for crowd motion and hard congestion, where a system evolves subject to an upper bound on local density~\cite{MauryRoudneffSantambrogioVenel2011,LeclercMerigotSantambrogioStra2020}.  Our question differs from these works in both its discrete one-dimensional setting and its objective: rather than studying equilibrium measures, interaction energies, or the time of a specified physical dynamics, we ask for the minimum total movement required to transform random initial locations into a feasible hard-core configuration.

\subsection{Lower Bound for the Uniform Distribution}
We begin by studying the case where~$\mathcal{D}=U\left(\Zn\right)$ is the uniform distribution over the cycle. 

For an integer $L\le n/2$ to be fixed later, we use the notion of \emph{cyclic intervals} throughout the proof.
For $s\in\Zn$, write
\[
        I_s=\{s,s+1,\ldots,s+L-1\}\subseteq\Zn.
\]
For any multi-set~$H_1,\ldots,H_m$ of home locations, let
\[
        X_s=\#\{i:H_i\in I_s\}
\]
be the number of home locations in $I_s$.
For~$x\in \mathbb{R}$, we denote by~$(x)_+ := \max(0,x)$ the positive part of~$x$.

\begin{lemma}[Interval overload lower bound]\label{lem:interval-overload}
For every realization of $H_1,\ldots,H_m$ and every distinct placement $Y_1,\ldots,Y_m$,
\[
        \sum_{i=1}^m |H_i-Y_i|_n
        \ge
        \sum_{s\in\Zn}(X_s-L)_+.
\]
\end{lemma}

\begin{proof}
Fix an interval $I_s$.
Only $L$ items can be finally placed inside $I_s$.
Therefore, if $X_s>L$, at least $X_s-L$ items whose homes lie in $I_s$ must be placed outside $I_s$:
\[
        (X_s-L)_+
        \le
        \#\{i:H_i\in I_s,\; Y_i\notin I_s\}.
\]
Summing over $s$ gives
\[
        \sum_s (X_s-L)_+
        \le
        \sum_{s=1}^{n} \#\{i:H_i\in I_s,\; Y_i\notin I_s\} = 
        \sum_{i=1}^m N_i,
\]
where $N_i$ is the number of length-$L$ cyclic intervals~$I_s$ that contain $H_i$ but not $Y_i$.

We claim that \(N_i\le |H_i-Y_i|_n\).  For a point \(x\in\mathbb Z_n\),
let
\[
        A_x=\{s\in\mathbb Z_n:x\in I_s\}
\]
be the set of starts of length-\(L\) intervals containing \(x\).  This is also a
cyclic interval of length \(L\).  If \(x'\) is adjacent to
\(x\), then \(A_x\setminus A_{x'}\) has size at most \(1\): shifting the
contained point by one step shifts the corresponding start interval by one
step, deleting at most one start and adding at most one start.

Now let \(d=|H_i-Y_i|_n\), and move from \(H_i\) to \(Y_i\) along a shortest
cyclic path
\[
        H_i=x_0,x_1,\ldots,x_d=Y_i .
\]
By the elementary inclusion
\[
        A_{x_0}\setminus A_{x_d}
        \subseteq
        \bigcup_{j=0}^{d-1}(A_{x_j}\setminus A_{x_{j+1}}),
\]
we get
\[
        N_i
        =
        |A_{H_i}\setminus A_{Y_i}|
        \le
        \sum_{j=0}^{d-1}|A_{x_j}\setminus A_{x_{j+1}}|
        \le d
        =
        |H_i-Y_i|_n .
\]
Therefore
\[
        \sum_s (X_s-L)_+
        \le
        \sum_i N_i
        \le
        \sum_i |H_i-Y_i|_n ,
\]
which proves the lemma.
\end{proof}

Consider a natural choice of parameter~$L=\Theta\left(\frac{1}{\varepsilon^2}\right)$ for load~$m=(1-\varepsilon)n$.
When~$\mathcal{D}$ is uniform, the expected number of home locations mapped into an interval~$I_s$ is~$(1-\varepsilon)L$ with standard deviation roughly~$\sqrt{L}$.
As~$\sqrt{L} \approx \varepsilon L$, we expect that with constant probability the interval~$I_s$ will overflow by~$\Theta\left(\sqrt{L}\right)$ elements.
We next spell out this quantitative lower bound for the uniform case.

\begin{lemma}[Uniform binomial overload]\label{lem:binomial-overload}
There are absolute constants $a,c,C>0$ such that the following holds.
Let $0<\eps<1/10$, let $n\eps^2\ge C$, let
\[
        m=\lfloor(1-\eps)n\rfloor,
        \qquad
        L=\left\lfloor \frac{a}{\eps^2}\right\rfloor,
\]
and let $Z\sim\Bin(m,L/n)$.
Then
\[
        \E[(Z-L)_+]\ge \frac{c}{\eps}.
\]
\end{lemma}

\begin{proof}
Choose $a>0$ sufficiently small; the choice will be fixed below.
Since $n\eps^2\ge C$ and $C$ is large enough, $1\ll L\le n/2$.
Let $q=L/n$.
Then
\[
        \mu:=\E [Z]=mq=(1-\eps)L\pm O(1),
\]
and
\[
        \sigma^2=\Var(Z)=mq(1-q)=\Theta(L),
\]
with absolute constants in the $\Theta(\cdot)$ notation.
Moreover,
\[
        L-\mu\le 2\eps L
        \le 2\sqrt a\,\sqrt L.
\]
Since $\sigma=\Theta(\sqrt L)$, by choosing $a$ small enough we ensure
\[
        L-\mu\le \sigma/4.
\]
By the Berry--Esseen theorem for binomial random variables, once $C$ is sufficiently large,
\[
        \Prb[Z\ge \mu+\sigma/2]\ge p
\]
for an absolute constant $p>0$.
On this event,
\[
        Z-L\ge \sigma/2-(L-\mu)\ge \sigma/4.
\]
Thus
\[
        \E[(Z-L)_+]\ge p\sigma/4=\Omega(\sqrt L)=\Omega(1/\eps).
\]
\end{proof}

The above implies our desired transport lower bound for the uniform distribution.

\begin{theorem}[Uniform transport lower bound]\label{thm:uniform-transport}
There is an absolute constant $c>0$ such that the following holds.
Let $m=\lfloor(1-\eps)n\rfloor$ and assume $n\eps^2$ is sufficiently large.
If $H_1,\ldots,H_m$ are independent uniform points of $\Zn$, then
\[
        \E\left[\OPT(H_1,\ldots,H_m)\right]
        \ge c\frac{n}{\eps}.
\]
\end{theorem}

\begin{proof}
Let $L=\lfloor a/\eps^2\rfloor$, with $a$ as in Lemma~\ref{lem:binomial-overload}.
For each $s\in\Zn$, $X_s\sim\Bin(m,L/n)$.
By Lemma~\ref{lem:interval-overload} and linearity of expectation,
\[
        \E\left[\OPT(H_1,\ldots,H_m)\right]
        \ge
        \sum_{s\in\Zn}\E[(X_s-L)_+]
        =
        n\E[(\Bin(m,L/n)-L)_+].
\]
The result follows from Lemma~\ref{lem:binomial-overload}.
\end{proof}

\subsection{Lower Bound for Any Distribution}
Next, we want to generalize the above for any distribution~$\mathcal{D}$ over~$\Zn$. 
We do so by proving that the quantity
\[
\E_{H\sim\mathcal{D}^m}\left[\sum_{s\in\Zn}(X_s-L)_+\right],
\]
which we used to lower bound the expectation of~$\OPT$, is minimized when~$\mathcal{D}$ is the uniform distribution. This would imply that the same lower bound holds for any distribution over~$\Zn$.
To do so, we use a convexity argument.

\begin{lemma}[Convexity of binomial overload]\label{lem:convexity}
Fix integers \(m\ge 0\) and \(L\ge 0\).  The function
\[
        g(p):=\E\left[(\Bin(m,p)-L)_+\right]
\]
is convex on \([0,1]\).
\end{lemma}

\begin{proof}
Let \(Z\sim\Bin(m,p)\), and write
\[
        g(p)=\sum_{k=0}^m \binom{m}{k}p^k(1-p)^{m-k}\cdot (k-L)_+ .
\]
This is a polynomial in \(p\), and hence differentiable, so it is enough to show that its derivative is non-decreasing on \([0,1]\).

We first prove the derivative formula.  For any function
\(f:\{0,\ldots,m\}\to\mathbb R\), define
\[
        F_f(p)=\E[f(\Bin(m,p))]
        =
        \sum_{k=0}^m f(k)\binom{m}{k}p^k(1-p)^{m-k}.
\]
Differentiating term by term gives
\[
\begin{aligned}
F_f'(p)
&=
\sum_{k=0}^m
f(k)\binom{m}{k}
\left(
k p^{k-1}(1-p)^{m-k}
-(m-k)p^k(1-p)^{m-k-1}
\right).
\end{aligned}
\]
Using
\[
        k\binom{m}{k}=m\binom{m-1}{k-1}
        \quad\text{and}\quad
        (m-k)\binom{m}{k}=m\binom{m-1}{k},
\]
we rewrite the two sums as
\[
\begin{aligned}
F_f'(p)
&=
m\sum_{k=1}^m
f(k)\binom{m-1}{k-1}
p^{k-1}(1-p)^{m-k}        \\
&\qquad
-
m\sum_{k=0}^{m-1}
f(k)\binom{m-1}{k}
p^k(1-p)^{m-1-k}.
\end{aligned}
\]
Changing variables \(j=k-1\) in the first sum yields
\[
\begin{aligned}
F_f'(p)
&=
m\sum_{j=0}^{m-1}
\bigl(f(j+1)-f(j)\bigr)
\binom{m-1}{j}p^j(1-p)^{m-1-j}  \\
&=
m\,\E_{Z\sim \Bin(m-1,p)}\left[
        f(Z+1)-f(Z)
      \right].
\end{aligned}
\]

Now apply this identity to
$
        f(k)=(k-L)_+.
$
For every integer \(j\in\{0,\ldots,m-1\}\),
\[
        f(j+1)-f(j)
        =
        (j+1-L)_+-(j-L)_+
        =
        \mathbf 1_{\{j\ge L\}} .
\]
Therefore,
\[
        g'(p)
        =
        m\,\Pr\big[\Bin(m-1,p)\ge L\big].
\]

As \(p\mapsto \Pr[\Bin(m-1,p)\ge L]\) is clearly non-decreasing, \(g'\) is
non-decreasing on \([0,1]\).  Consequently \(g\) is convex.
\end{proof}

\begin{lemma}[Uniform distribution minimizes expected interval overload]
\label{lem:uniform-minimizes-overload}
Fix integers \(m\ge 0\) and \(L\ge 0\).  Let \(\mathcal D\) be a distribution
on \(\Zn\), let \(H_1,\ldots,H_m\sim \mathcal D\) independently. Then
\[
        \E_{H\sim\mathcal D^m}\left[\sum_{s\in\Zn}(X_s-L)_+\right]
        \ge
        \E_{H\sim U(\Zn)^m}
        \left[\sum_{s\in\Zn}(X_s-L)_+\right].
\]
\end{lemma}

\begin{proof}
For each \(s\in\Zn\), write
\[
        p_s=\Pr_{H\sim\mathcal D}[H\in I_s].
\]
Since \(H_1,\ldots,H_m\) are sampled independently from \(\mathcal D\), we have
\[
        X_s\sim \Bin(m,p_s).
\]
Therefore,
\[
        \E_{H\sim\mathcal D^m}\left[\sum_{s\in\Zn}(X_s-L)_+\right]
        =
        \sum_{s\in\Zn} g(p_s),
\]
where~$g(p)$ is defined as in Lemma~\ref{lem:convexity}.
By Lemma~\ref{lem:convexity}, \(g\) is convex.

Next,
\[
        \frac1n\sum_{s\in\Zn}p_s
        =
        \frac1n\sum_{s\in\Zn}\Pr_{H\sim\mathcal D}[H\in I_s].
\]
Interchanging the sum and the probability measure gives
\[
        \frac1n\sum_{s\in\Zn}p_s
        =
        \frac1n\E_{H\sim\mathcal D}\left[
        \#\{s\in\Zn:H\in I_s\}
        \right].
\]
Every cell of \(\Zn\) belongs to exactly \(L\) cyclic intervals of length \(L\).
Hence
\[
        \frac1n\sum_{s\in\Zn}p_s
        =
        \frac Ln.
\]
By Jensen's inequality,
\[
        \frac1n\sum_{s\in\Zn}g(p_s)
        \ge
        g\left(\frac1n\sum_{s\in\Zn}p_s\right)
        =
        g\left(\frac Ln\right).
\]
Thus
\[
        \E_{H\sim\mathcal D^m}\left[\sum_{s\in\Zn}(X_s-L)_+\right]
        \ge
        n g\left(\frac Ln\right).
\]
By the definition of \(g\), this is
\[
        \E_{H\sim\mathcal D^m}\left[\sum_{s\in\Zn}(X_s-L)_+\right]
        \ge
        n\cdot
        \E\left[
        \left(\Bin\left(m,\frac Ln\right)-L\right)_+
        \right].
\]

Finally, if \(\mathcal D\) is the uniform distribution on \(\Zn\), then
\(p_s=L/n\) for every \(s\in\Zn\), and equality holds in the expression above.
Therefore the expected interval overload is minimized by the uniform
distribution.
\end{proof}

This allows us to repeat the proof of Theorem~\ref{thm:uniform-transport} for a general distribution~$\mathcal{D}$.
\begin{theorem}[Distribution-free i.i.d. transport lower bound]
\label{thm:general-transport}
There is an absolute constant \(c>0\) such that the following holds.
Let \(m=\lfloor(1-\eps)n\rfloor\), and assume \(n\eps^2\) is sufficiently large.
For every distribution \(\mathcal D\) on \(\Zn\), if
\[
        H_1,\ldots,H_m\sim \mathcal D
\]
independently, then
\[
        \E_{H\sim\mathcal D^m}\left[\OPT(H_1,\ldots,H_m)\right]
        \ge
        c\frac{n}{\eps}.
\]
\end{theorem}

\begin{proof}
Set
\[
        L=\left\lfloor \frac{a}{\eps^2}\right\rfloor,
\]
where \(a>0\) is the absolute constant from Lemma~\ref{lem:binomial-overload}.
For each \(s\in\Zn\), let
\[
        X_s=\#\{i:H_i\in I_s\}.
\]
By Lemma~\ref{lem:interval-overload}, for every realization of
\(H_1,\ldots,H_m\),
\[
        \OPT(H_1,\ldots,H_m)
        \ge
        \sum_{s\in\Zn}(X_s-L)_+ .
\]
Taking expectations gives
\[
        \E_{H\sim\mathcal D^m}\left[\OPT(H_1,\ldots,H_m)\right]
        \ge
        \E_{H\sim\mathcal D^m}
        \left[\sum_{s\in\Zn}(X_s-L)_+\right].
\]
Applying Lemma~\ref{lem:uniform-minimizes-overload},
\[
        \E_{H\sim\mathcal D^m}\left[\OPT(H_1,\ldots,H_m)\right]
        \ge
        n\cdot
        \E\left[
        \left(\Bin\left(m,\frac Ln\right)-L\right)_+
        \right].
\]
Thus by Lemma~\ref{lem:binomial-overload} we have
\[
        \E_{H\sim\mathcal D^m}\left[\OPT(H_1,\ldots,H_m)\right]
        \ge
        c\frac{n}{\eps}
\]
for an absolute constant \(c>0\).
\end{proof}

\section{Hashing Lower Bounds}\label{sec:hashing-lb}

In this section we prove all of our hashing lower bounds, which are corollaries of the transport lower bounds from Section~\ref{sec:transport}.
In particular, we show that any open-addressing algorithm without reordering can't have expected locality~$o(1/x^2)$ for all loads~$(1-x)$, must have amortized expected locality~$\Omega(1/\varepsilon)$ in a sequence of~$(1-\varepsilon)n$ insertions, and that page size~$B=\Omega(1/\varepsilon^2)$ is necessary for expected~$1+o(1)$ page accesses at load~$(1-\varepsilon)$ in the external-memory model. 
All of these lower bounds match the upper bounds achieved by linear probing.

For our lower bounds, we make only minimal assumptions about the hash algorithm. In particular, we do not assume the algorithm is greedy or uses a fixed sequence of shifts. 
All assumptions we make on the hash table algorithm are as follows:
\begin{itemize}
    \item \textbf{Open Addressing without Reordering:} When a new element is inserted, the algorithm eventually places it in some empty cell~$i\in[n]$ in the length~$n$ table~$T$. The element may not move to other cells in subsequent operations.
    \item \textbf{Stateless First Probe:} When an element~$x$ is inserted, the index of the first cell probed in~$T$ is a possibly-randomized variable that depends only on~$x$ and any auxiliary information set during the table's initialization (such as the table size~$n$, the hash function, etc.), but not on the previously inserted elements.
    \item \textbf{Sufficiently Large Universe:} The size of the universe~$\mathcal{U}=\mathcal{U}_n$ of elements that can be inserted to the table is~$\omega(n^2)$.
\end{itemize}

We note that the second assumption is necessary for meaningful bounds, as otherwise the algorithm could remember which cells are still empty and pick the first probe location within this set. A lower bound on the universe size is also crucial, as if for example~$|\mathcal{U}_n|=n$ one could simply use the identity function as the hash function to guarantee no collisions.

All of our lower bounds are obtained by considering the following natural input sequence: For~$m$ times, we insert a uniformly random universe element~$u$ from~$\mathcal{U}$ into the hash table.
Since \(m\le n\) and~$|\mathcal U|=\omega(n^2)$, these keys are all distinct with probability
\[
        1-O(m^2/|\mathcal U_n|)=1-o(1).
\]
Thus the hard distribution may be interpreted as an ordinary sequence of
insertions of distinct keys, up to a vanishing-probability event.  
As \(u_1,\ldots,u_m\) are independent uniform elements of \(\mathcal U_n\), their \emph{home locations}, the index of the first probed cell when inserting each,
\[
        H_i=h(u_i)
\]
are independent samples from the same induced distribution~$\mathcal D$ on \(\Zn\), after
conditioning on the algorithm's initialization.  
This property is why the distribution-free transport theorem applies.

\begin{theorem}[No subquadratic locality]
\label{thm:no-all-load-hashing}
No open-addressing algorithm without reordering can have expected insertion locality
\(o(x^{-2})\) at every load \(1-x\).  More precisely, there is no locally
bounded function \(f:(0,1]\to\mathbb R_+\) with
\[
        f(x)=o(x^{-2})
        \qquad\text{as }x\rightarrow 0^+
\]
such that, for every sufficiently large \(n\), every insertion performed at
load \(1-x\) the expected insertion locality satisfies
\[
        \E[R]\le f(x).
\]
\end{theorem}

Before proving Theorem~\ref{thm:no-all-load-hashing}, we state and prove the amortized lower bound from which
it follows.  The following statement allows the final load \(1-\eps\) to be known to the
algorithm in advance.

\begin{theorem}[Amortized locality lower bound]
\label{thm:amortized-hashing-lb}
Let \(m=\lfloor(1-\eps)n\rfloor\), and assume \(n\eps^2\) is sufficiently large.
For every open-addressing algorithm without reordering, even one that is given
\(n\), \(m\), and \(\eps\) in advance, the random insertion sequence
satisfies
\[
        \sum_{i=1}^{m}\E[R_i]
        =
        \Omega\left(\frac{n}{\eps}\right).
\]
Consequently, the amortized expected insertion locality is \(\Omega(1/\eps)\).
\end{theorem}

\begin{proof}
Fix the algorithm's state after initialization and let \(\mathcal D\) be the
induced distribution of the first probe location when the inserted element \(u\) is uniform in \(\mathcal U\).  
For the hard insertion sequence, the first probe locations~$H_1,\ldots,H_m$
are independent samples from \(\mathcal D\).

Since the algorithm stores the inserted elements in distinct cells, it produces distinct final locations \(Y_1,\ldots,Y_m\).  
Each insertion must inspect the cell in which it finally stores its element, and therefore
\[
        R_i\ge |H_i-Y_i|_n
\]
for every \(i\).  Hence
\[
        \sum_{i=1}^m R_i
        \ge
        \sum_{i=1}^m |H_i-Y_i|_n
        \ge
        \OPT(H_1,\ldots,H_m).
\]
Taking expectation over the hard sequence and over the algorithm's remaining
randomness after initialization, and then applying Theorem~\ref{thm:general-transport} to the
arbitrary distribution \(\mathcal D\), gives
\[
        \sum_{i=1}^{m}\E[R_i]
        \ge
        \E_{H\sim\mathcal D^m}\left[\OPT(H_1,\ldots,H_m)\right]
        =
        \Omega\left(\frac{n}{\eps}\right).
\]
Finally, averaging over the initialization randomness of the algorithm gives the
same bound unconditionally.
\end{proof}

We next show that the amortized lower bound of Theorem~\ref{thm:amortized-hashing-lb} implies the worst-case lower bound of Theorem~\ref{thm:no-all-load-hashing} by a simple integration argument.

\begin{proof}[Proof of Theorem~\ref{thm:no-all-load-hashing}]
Assume, toward a contradiction, that such an algorithm and function \(f\) exist.
Since \(f(x)=o(x^{-2})\), for every \(\eta>0\) there is a sufficiently small
constant \(\delta_0>0\) such that
\[
        f(x)\le \eta x^{-2}
        \qquad\text{for all }0<x\le \delta_0 .
\]
Since \(f\) is locally bounded, let
\[
        M_0=\sup_{x\in[\delta_0,1]} f(x)<\infty .
\]
Run the algorithm for
\[
        m=\lfloor(1-\eps)n\rfloor
\]
insertions, where \(\eps<\delta_0\) and \(n\eps^2\) is sufficiently large.  Denote by
\(x_i=1-(i-1)/n\), so that~$1-x_i$ is the load before insertion \(i\).  The insertions with
\(x_i\ge\delta_0\) contribute at most \(M_0 n\) to the expected total locality.
For the remaining insertions,
\[
\begin{aligned}
        \sum_{i \text{ s.t. } x_i<\delta_0}\E[R_i]
        &\le
        \eta\sum_{i \text{ s.t. } x_i<\delta_0} x_i^{-2}  \\
        &\le
        C\eta n\int_{\eps}^{\delta_0} \frac{dx}{x^2}  \\
        &\le
        C\eta\frac{n}{\eps},
\end{aligned}
\]
for an absolute constant \(C\).  Choosing \(\eta\) sufficiently small and then
\(\eps\) sufficiently small gives
\[
        \sum_{i=1}^{m}\E[R_i]
        <
        c\frac{n}{\eps},
\]
where \(c\) is the constant from Theorem~\ref{thm:amortized-hashing-lb}.  This
contradicts Theorem~\ref{thm:amortized-hashing-lb}.
\end{proof}

We note that linear probing is optimal with respect to all-load locality.  Indeed, at load
\(1-x\), linear probing has expected insertion probe count (and hence also locality) \(\Theta(x^{-2})\).
Theorem~\ref{thm:no-all-load-hashing} shows that no immutable open-addressing
scheme can improve this asymptotic locality profile at all loads, even though
linear probing is suboptimal in probe count.  The amortized lower bound of
Theorem~\ref{thm:amortized-hashing-lb} is also tight for linear probing, since
summing \(\Theta(x^{-2})\) over the loads \(x\in[\eps,1]\) gives
\(\Theta(n/\eps)\).

\subsection{External-memory consequence}

The locality lower bound also resolves a long-standing open problem for a classical
threshold in external-memory hashing, for open addressing tables without reordering -- the framework frequently used in practice due to its simple implementation.
In the external-memory model, memory is
partitioned into pages, or blocks, of \(B\) consecutive cells, and the cost of an
operation is measured by the number of pages it needs to bring from memory
\cite{AggarwalVitter1988}. This cleanly models that the time bottleneck in practice is often the cost of memory access, especially in settings like hash tables in which the number of operations following each query is tiny.  
A central goal is to obtain \(1+o(1)\) page accesses per operation at high load, meaning that an operation is served almost always within the page containing its first probe.

A separate line of work optimizes the tradeoff between buffered insertion I/Os
and query I/Os in external-memory dictionaries
\cite{IaconoPatrascu2012,ConwayFarachColtonShilane2018}.
These structures use batching and rebuilding, and thus address a different
notion of insertion locality from the immutable open-addressing setting studied
here.

The classical threshold associated with high-load external-memory hashing is
\(B=\Theta(1/x^2)\).  Jensen and Pagh gave an external-memory hash table with
expected lookup cost \(1+O(1/\sqrt B)\), amortized expected update cost
\(1+O(1/\sqrt B)\) after a lookup, and space usage \(1+O(1/\sqrt B)\) times
optimal~\cite{JensenPagh2008}.  Equivalently, this gives \(1+o(1)\) expected
page accesses at load \(1-O(1/\sqrt B)\), or \(B=\Theta(1/x^2)\) at load
\(1-x\).

This \(1/x^2\) threshold has been viewed as a major barrier.  Bender, Kuszmaul,
and Kuszmaul~\cite{BenderKuszmaulKuszmaul2021} formulate the corresponding ``space-efficient external-memory hashing'' problem as achieving \(1+o(1)\) expected block accesses at load
\(1-x\), and present a significant improvement to the above tradeoff by getting  \(B=\Theta(1/x)\) when the algorithm is allowed to reorder elements and bounds are amortized.
They state that, before their work, and still for all algorithms without reordering, the best known constructions only achieved such a guarantee when \(B=\Omega(1/x^2)\).
Our result shows this was not a coincidence and supplies the first lower bound for algorithms without reordering, showing that \(B=\Omega(1/x^2)\) is in fact necessary.
We denote by the \emph{page span} of
an operation the number of consecutive pages intersected by the interval
between the first probed cell and the farthest cell inspected or used; We note that an operation that only accesses one page also has a page span of one.

\begin{theorem}[External-memory locality lower bound]
\label{thm:external-memory}
Let \(x=1/\eps\).  Consider any open-addressing algorithm without reordering satisfying
the assumptions of Theorem~\ref{thm:no-all-load-hashing}.  If the algorithm has
expected page span \(1+o(1)\) at every load \(1-1/x\), then
\[
        B=\Omega(x^2).
\]
Equivalently, at load \(1-\eps\), page size
\[
        B=\Omega(1/\eps^2)
\]
is necessary for \(1+o(1)\) expected page span in all cache-oblivious schemes.
\end{theorem}

\begin{proof}
Let \(R\) denote the locality of an insertion, namely the distance from the
first probed cell to the farthest cell inspected or used.  If the table is
partitioned into pages of \(B\) consecutive cells, then an operation of locality
\(R\) has page span \(\Omega(R/B)\).  Hence
\[
        \E[\mathrm{page\ span}]
        \ge
        \Omega\left(\frac{\E \left[R\right]}{B}\right).
\]
If \(B=o(x^2)\) and the expected page span were \(1+o(1)\) at every load
\(1-1/x\), then the expected locality would be \(o(x^2)\) at every such load.
This contradicts Theorem~\ref{thm:no-all-load-hashing}.  Therefore \(B=\Omega(x^2)\).
\end{proof}

\section{Hashing Upper Bounds}\label{sec:upper}

All locality lower bounds in Section~\ref{sec:hashing-lb} are matched by linear probing.
On the other hand, as is clearly established by now, the number of probes made by linear probing is suboptimal.
This leads to a natural question:
\begin{question}
    Can we optimize both the locality and the number of probes simultaneously?
\end{question}
An intriguing candidate to obtain this is \emph{quadratic probing}, a simple variant of the greedy linear probing in which we replace the probe sequence~$h_i(x)=h(x)+i$ with~$h_i(x)=h(x)+i^2$~\cite{Maurer1968}.
Despite its apparent simplicity, analyzing the expected probe count of this algorithm -- or in fact, any fixed sequence of shifts other than linear probing -- remains frustratingly elusive.
Only recently, Kuszmaul and Xi~\cite{KuszmaulXi2024}
proved that there exists a constant load factor for which the expected probe count is constant.
Nonetheless, it is frequently conjectured that quadratic probing achieves an expected probe count~$O(1/\varepsilon)$ -- which is optimal among greedy algorithms~\cite{Yao1985, Radke1970, HopgoodDavenport1972, Ecker1974, Batagelj1975, Weiss2000, CormenLeisersonRivestStein2009, KuszmaulXi2024}.
If that conjecture holds, then the combination of expected probe count linear in~$1/\varepsilon$ with the quadratic growth of distances between each probe and the home location -- yields, in some sense, a simultaneous optimization of both probe count and locality.

Our main goal in this section is to establish, unconditionally, the same behavior that is described above conditioned on quadratic probing truly having an optimal probe count.
To this end we present a natural greedy probing scheme which works essentially as follows: we guess the slack~$\varepsilon$, by enumeration over~$\varepsilon=1/2,1/4,1/8,\ldots$, and for each guess we attempt insertion in~$\Theta(1/\varepsilon)$ random cells sampled uniformly from all those of distance at most~$\Theta(1/\varepsilon^2)$ from the home location.
Intuitively, this should achieve the desired bounds, as we can hope that an interval of length~$L$ would have only~$(1-\varepsilon)L\pm O(\sqrt{L})$ occupied cells; If so, an interval of length~$\Theta(1/\varepsilon^2)$ would likely contain an~$\varepsilon$-fraction of empty cells which would be hit by our set of samples.
Following on that intuition though is not trivial: we need to show that the distribution of the number of empty cells in such an interval behaves sufficiently similarly to a Binomial distribution -- despite the previous elements being inserted via this weird and structured sequence of probes. 

Our main technical contribution in this section, then, is establishing a generic bound on the variance of the amount of empty cells within any subset of cells, for any open-addressing algorithm that is \emph{symmetric} (or \emph{translation invariant}).
This generic bound also immediately yields an upper bound of~$O(\log n / \eps^2)$ for the expected number of probes in \emph{any} symmetric scheme and any load~$1-\varepsilon$, including all fixed-shift sequence schemes.

We remark that the above still lacks a desired property: if the expected probe count is~$\E\left[T\right]=O\left(1/\varepsilon\right)$ and the distance of each probe from the home location grows quadratically~$R=O(T^2)$, then while we do get that the locality is~$O(1/\varepsilon^2)$ with arbitrarily high constant probability -- we do not automatically get that the \emph{expected} locality is. This is because a bound on~$\E\left[T\right]$ is insufficient to imply a bound on~$\E\left[T^2\right]$.

\subsection{Symmetric greedy probing mechanisms}

A greedy probing mechanism is \emph{symmetric} if, for each inserted key, the entire probe sequence is sampled before the insertion from a distribution that is independent of the current table with shifts that are invariant under translations of the home location.
Equivalently, if the home is $H\in\Zn$, the probe sequence has the form
\[
        H+S_0,H+S_1,H+S_2,\ldots,
        \qquad S_0=0,
\]
where the joint distribution of $(S_0,S_1,S_2,\ldots)$ is fixed in advance.
Different keys use independent copies of this random sequence, and home locations are independent uniform points of $\Zn$.
The insertion stores the key in the first empty probed cell.

This definition allows arbitrary dependence among the shifts of one key.
It includes linear probing, double hashing with a random step, uniform probing, and the expanding-window scheme below.
The only properties used in the density proof are translation invariance, independence between keys, and greediness with no reordering.

\subsection{A variance bound for all symmetric schemes}

The following lemma is the main technical input for the greedy upper bound.
As we consider translation-invariant schemes, the expected number of occupied cells in a fixed interval is simply the current load multiplied by the size of the interval. The following Lemma shows that, in all symmetric schemes, the variance of that number of occupied cells is at most, asymptotically, the variance we get in linear or uniform probing.
In other words, it says that no symmetric probing mechanism can create more than linear variance in the number of occupied cells seen by a fixed interval. 

\begin{lemma}[Interval variance for symmetric probing]\label{lem:symmetric-variance}
Consider any symmetric probing mechanism after $t$ insertions, and let $O_t\subseteq\Zn$ be the occupied set.
For every fixed set $I\subseteq\Zn$, let
\[
        X_I=|O_t\cap I|.
\]
Then
\[
        \Var(X_I)\le \frac{t}{n}|I|\le |I|.
\]
Consequently, if $t\le(1-\eps)n$ and $Y_I=|I|-X_I$ is the number of empty cells in $I$, then
\[
        \E\left[Y_I\right]\ge \eps |I|
        \quad \text{and}\quad
        \Var(Y_I)\le |I|.
\]
\end{lemma}

\begin{proof}
Let $Z_i$ denote the full random data of the $i$th key: its home location and its complete probe sequence.
Thus $X_I=f(Z_1,\ldots,Z_t)$ for a deterministic function $f$.
Let $Z_i'$ be an independent copy of $Z_i$, and let $X_I^{(i)}$ be the interval count after replacing $Z_i$ by $Z_i'$ and leaving all other keys unchanged.
By the Efron-Stein concentration inequality (proof may be found at~\cite{boucheron2004concentration}),
\[
        \Var(X_I)
        \le
        \frac12\sum_{i=1}^t
        \E\left[\left(X_I-X_I^{(i)}\right)^2\right].
\]

We compare the two executions that differ only in the data of the~$i$-th inserted element.
Immediately before element $i$ is inserted the two tables are identical.
Immediately after inserting element $i$, the two occupied sets either are identical or differ by one swap: one cell occupied in the first execution and one cell occupied in the second execution.
We claim that this invariant persists through all later insertions.
Indeed, suppose two occupied sets $O,O'$ of the same size differ by at most one swap, say $O\setminus O'=\{a\}$ and $O'\setminus O=\{b\}$.
Expose the same probe sequence for the next key in both executions.
If the first empty probed cell is the same in both tables, the swap remains unchanged.
Otherwise the first discrepancy encountered by the probe sequence is one of $a,b$; after inserting the key, the discrepancy either disappears or moves to the cell chosen by the other execution.
In all cases the two new occupied sets again differ by at most one swap.

Thus, after all $t$ insertions, the two final occupied sets differ by at most two cells.
When they differ, write these two cells as $A_i$ and $B_i$, with $A_i$ occupied only in the first execution and $B_i$ occupied only in the resampled execution.
Then
\[
        |X_I-X_I^{(i)}|\le 1,
        \qquad
        (X_I-X_I^{(i)})^2
        \le \one[A_i\in I]+\one[B_i\in I].
\]
By translation invariance of the whole sequence, each of $A_i$ and $B_i$, conditional on existing, is uniformly distributed over $\Zn$.
Therefore
\[
        \Prb[A_i\in I]\le \frac{|I|}{n},
        \qquad
        \Prb[B_i\in I]\le \frac{|I|}{n}.
\]
It follows that
\[
        \E\left[\left(X_I-X_I^{(i)}\right)^2\right]
        \le \frac{2|I|}{n}.
\]
Plugging this into the concentration inequality gives
\[
        \Var(X_I)
        \le
        \frac12\sum_{i=1}^t\frac{2|I|}{n}
        =\frac{t}{n}|I|.
\]
Finally, $Y_I=|I|-X_I$, so $\Var(Y_I)=\Var(X_I)$.
Since the occupied set is translation invariant, $\E \left[X_I\right]=t|I|/n$, and hence $\E\left[ Y_I\right]=(1-t/n)|I|\ge\eps |I|$ when $t\le(1-\eps)n$.
\end{proof}

The above bound on the variance is sufficient to give an upper bound on the probability an interval is significantly over-occupied. 

\begin{corollary}[Density in intervals]\label{cor:critical-density}
Let a symmetric probing mechanism be at load at most $1-\eps$.
Let $H$ be an independent uniform point of $\Zn$, and let
$
        I=[H,H+L)
$
be a cyclic interval of length $L$.
For every $0<\eta<1$,
\[
        \Prb\left[|I\setminus O_t|<\eta\eps L\right]
        \le
        \frac{1}{(1-\eta)^2\eps^2L}.
\]
In particular, if $L=C/\eps^2$, then
\[
        \Prb\left[|I\setminus O_t|\ge \eta\eps L\right]
        \ge
        1-\frac{1}{(1-\eta)^2C}.
\]
\end{corollary}

\begin{proof}
Condition on $H$ or, equivalently by translation invariance, fix the interval $I$.
Let $Y_I=|I\setminus O_t|$ be the number of free cells in the interval~$I$.
By Lemma~\ref{lem:symmetric-variance},
\[
        \E \left[Y_I\right]\ge \eps L,
        \qquad
        \Var(Y_I)\le L.
\]
Therefore Chebyshev gives
\[
        \Prb[Y_I<\eta\eps L]
        \le
        \Prb\bigl[|Y_I-\E \left[Y_I\right]|>(1-\eta)\eps L\bigr]
        \le
        \frac{L}{(1-\eta)^2\eps^2L^2}.
\]
\end{proof}

\begin{remark}[A weak bound for any fixed-offset probing]
\label{rem:fixed-offset-logn}
Let \(s_1,s_2,\ldots\) be a fixed offset sequence.  Let \(T\) be the insertion time at load
\(1-\varepsilon\).  Applying Lemma~\ref{lem:symmetric-variance} to the translated set of
the first \(k\) probed locations gives
\[
        \Pr[T>k]\le O\left(\frac{1}{\varepsilon^2 k}\right).
\]
Indeed, if \(S_k=\{s_1,\ldots,s_k\}\) and \(E_t\) is the set of empty cells,
then
\[
        T>k
        \quad\Longrightarrow\quad
        E_t\cap(h+S_k)=\emptyset,
\]
while Lemma~\ref{lem:symmetric-variance} gives
\[
        \E\left[ |E_t\cap(h+S_k)|\right]=\varepsilon k,
        \qquad
        \Var(|E_t\cap(h+S_k)|)\le k,
\]
and Chebyshev's inequality gives the above tail bound.  Consequently, 
\[
        \E \left[T\right]
        =
        \sum_{k\ge0}\Pr[T>k]
        \le
        \sum_{k=1}^{n}
        O\left(\frac{1}{\varepsilon^2 k}\right)
        =
        O(\varepsilon^{-2}\log n).
\]

To the best of our knowledge, this simple all-load consequence was not
previously stated.  It is much weaker than the constant-load geometric-tail
theorem of Kuszmaul and Xi~\cite{KuszmaulXi2024}, but unlike that theorem it applies at every load and also for randomized shift sequences. 
\end{remark}

\subsection{Optimal expected probe count with quadratic distance growth}

For $j=0,1,2,\ldots$, let
\[
        \Delta_j=2^{-j-1},
        \qquad
        W_j=\left\lceil \frac{A}{\Delta_j^2}\right\rceil,
        \qquad
        S_j=\left\lceil \frac{a}{\Delta_j}\right\rceil,
\]
where $A$ and $a$ are sufficiently large absolute constants.
During insertion, we uniformly sample a home location~$h(x)$.
Then, in phase $j$, we sample $S_j$ fresh cells uniformly from the interval
\[
        [h(x),h(x)+W_j)
        =\{h(x),h(x)+1,\ldots,h(x)+W_j-1\}\subseteq\Zn.
\]
The insertion probes the sampled cells in their sampled order until either an empty sampled cell is found or all $S_j$ sampled cells have been inspected.
If an empty sampled cell is found, the key is stored there; otherwise the insertion proceeds to phase $j+1$.
Once $W_j\ge n$, the algorithm samples uniformly from the whole table until it finds an empty cell.

\begin{algorithm}[h]
\caption{Expanding-window greedy probing}
\label{alg:expanding-window}
\begin{algorithmic}[1]
\Require key $x$ with home location $h(x)$
\For{$j=0,1,2,\ldots$}
    \State $\Delta_j\gets 2^{-j-1}$
    \State $W_j\gets \min\{n,\lceil A/\Delta_j^2\rceil\}$, $S_j\gets \lceil a/\Delta_j\rceil$
    \If{$W_j=n$}
        \State Probe fresh uniformly random table cells until an empty cell is found; store $x$ there.
    \Else
        \State Sample $S_j$ fresh uniformly random cells in $[h(x),h(x)+W_j)$.
        \State Probe the sampled cells in sampled order, stopping if an empty cell is found.
        \If{an empty cell is found}
            \State Store $x$ there.
        \EndIf
    \EndIf
\EndFor
\end{algorithmic}
\end{algorithm}

This is a symmetric greedy probing mechanism: the distribution of the whole sequence of shifts between the probes and the home location is fixed in advance and is independent of the current table or the home location itself.

\begin{lemma}[One-phase failure bound]\label{lem:one-phase-failure}
Suppose the current load is at most $1-\eps$, and consider a phase $j$ with $W_j<n$ and $\Delta_j\le \eps$.
Then the probability that phase $j$ fails is at most
\[
        \frac{4\Delta_j^2}{A\eps^2}+\exp\!\left(-\frac{a\eps}{4\Delta_j}\right),
\]
provided $A$ is larger than an absolute constant.
\end{lemma}

\begin{proof}
Let $I=[H,H+W_j)$ be the phase-$j$ window for the fresh home $H$.
Apply Corollary~\ref{cor:critical-density} with $\eta=1/2$ and $L=W_j$.
Since $W_j\ge A/\Delta_j^2$, the probability that $I$ contains fewer than $(\eps/2)W_j$ empty cells is at most
\[
        \frac{4}{\eps^2W_j}
        \le
        \frac{4\Delta_j^2}{A\eps^2}.
\]
On the complementary event, a uniformly sampled cell from $I$ is empty with probability at least $\eps/2$.
The probability that all $S_j\ge a/\Delta_j$ sampled cells are occupied is at most
\[
        (1-\eps/2)^{S_j}
        \le
        \exp(-\eps S_j/2)
        \le
        \exp\!\left(-\frac{a\eps}{2\Delta_j}\right).
\]
Weakening the constant in the exponent gives the stated bound.
\end{proof}

\begin{lemma}[Quadratic locality of the sequence]\label{lem:quadratic-sequence}
For Algorithm~\ref{alg:expanding-window}, every execution satisfies
\[
        R\le C T^2
\]
for an absolute constant $C$, where $T$ is the number of cell probes made by the operation and $R$ is the locality radius reached by the operation.
\end{lemma}

\begin{proof}
If the operation terminates before the full-table phase, let $j$ be the last phase reached.
The radius is at most $W_j=O(\Delta_j^{-2})$.
If $j=0$, then $W_j=O(1)$ and $T\ge1$.
If $j>0$, then reaching phase $j$ means that all previous phases failed, so the operation already made
\[
        \sum_{r<j}S_r=\Omega(\Delta_j^{-1})
\]
probes before phase $j$ began.
Thus in all cases $T=\Omega(\Delta_j^{-1})$, and hence $R=O(T^2)$.
If the operation reaches the full-table phase, then at the first such phase $W_j=n$, so $\Delta_j^{-1}=\Omega(\sqrt n)$.
The preceding phases have already made $\Omega(\sqrt n)$ probes, and the radius is at most $n$.
Again $R=O(T^2)$.
\end{proof}

\begin{theorem}[Greedy probe bound]\label{thm:greedy-upper}
At every load at most $1-\eps$, Algorithm~\ref{alg:expanding-window} satisfies
\[
        \E[T]=O(1/\eps).
\]
Moreover its probe sequence has quadratic locality, $R=O(T^2)$ deterministically.
\end{theorem}

\begin{proof}
Let $j_\eps$ be the first phase with $\Delta_{j_\eps}\le\eps$.
The total number of probes spent before phase $j_\eps$ is
\[
        \sum_{r<j_\eps}S_r=O(1/\eps),
\]
and the locality radius reached before that phase is $O(1/\eps^2)$.

For $k\ge0$ such that $W_{j_\eps+k}<n$, we have
\[
        \Delta_{j_\eps+k}=\Theta(\eps 2^{-k}),
        \qquad
        S_{j_\eps+k}=O(2^k/\eps),
        \qquad
        W_{j_\eps+k}=O(4^k/\eps^2).
\]
By Lemma~\ref{lem:one-phase-failure}, after choosing $A$ and $a$ sufficiently large,
\begin{equation}\label{eq:phase-fail-variance}
        \Prb[\text{phase }j_\eps+k\text{ fails}]
        \le
        C4^{-k}+\exp(-c2^k)
\end{equation}
for absolute constants $C,c>0$.
Consequently, for $k\ge1$,
\[
        \Prb[\text{the insertion reaches phase }j_\eps+k]
        \le
        C4^{-k}+\exp(-c2^k),
\]
after changing constants.

The expected number of probes in phases with $W_j<n$ is therefore
\[
        O(1/\eps)
        +
        \sum_{k\ge1}O(2^k/\eps)\bigl(C4^{-k}+\exp(-c2^k)\bigr)
        =O(1/\eps).
\]
If the algorithm reaches the first phase with $W_j=n$, it samples uniformly from the whole table until success, using an additional expected $O(1/\eps)$ probes.
This contribution is $O(1/\eps)$ even if the full-table phase is reached with probability one.
This proves the expected probe bound.
The deterministic quadratic relation is Lemma~\ref{lem:quadratic-sequence}.
\end{proof}

\section{Deamortized Upper Bound for Known Load}\label{sec:known-upper}

The two locality lower bounds proved in this paper describe two different restrictions:
Theorem~\ref{thm:no-all-load-hashing} shows that no immutable algorithm can
have locality $o(1/\varepsilon^2)$ at every load $1-\varepsilon$, whereas
Theorem~\ref{thm:amortized-hashing-lb} gives an amortized  lower bound of
$\Omega(1/\eps)$ after averaging over all insertions up to load
$1-\eps$, even when this final load is known in advance to the algorithm.  Both statements are
matched by linear probing: its expected cost at any slack $x$ is
$\Theta(1/x^2)$, while
\[
 \frac{1}{n}\sum_{t=0}^{\lfloor(1-\eps)n\rfloor-1}
        \Theta\left(\frac{1}{(1-t/n)^2}\right)
        =\Theta(1/\eps).
\]
This leaves a natural deamortization question.  If the number of inserted
elements $(1-\eps)n$ is known in advance, can some of the cost of the last
insertions be shifted to earlier ones, so that no individual insertion pays
the quadratic cost?  We show that this can be done, up to logarithmic factors.
The same guarantee holds for searches.

\begin{theorem}[Known-load upper bound]\label{thm:known-target-upper}
Fix $0<\eps<1/10$ in advance.  There is an immutable open-addressing scheme
such that every insertion before load $1-\eps$, and every successful or
unsuccessful search at load at most $1-\eps$, satisfies
\[
        \E[T],\ \E[R]
        =
        O\left(\frac{\log^3(1/\eps)}{\eps}\right).
\]
\end{theorem}

The construction is hierarchical, but we first describe a two-level version that captures the main primitive used by the algorithm:
We partition the table into disjoint consecutive \emph{super-blocks} of size~$\approx 1/\varepsilon^2$, and partition each of them to \emph{mini-blocks} of size~$\approx 1/\varepsilon^{3/2}$.
When a new element is inserted it flips a coin: with probability~$1-\Theta(\sqrt{\varepsilon})$ we try to insert it into a uniformly chosen mini-block, otherwise (with probability~$\Theta(\sqrt{\varepsilon})$) or if the mini-block from the first case was full, we will expand our search to the entire super-block the inserted element is mapped into.
As only a~$(1-\Theta(\sqrt{\varepsilon}))$ fraction of the elements attempt insertion into their mini-block, we expect they would have sufficient space: this insertion is similar to linear-probing with slack~$\Theta(\sqrt{\varepsilon})$, which requires looking only at~$\Theta({\varepsilon}^{-1})$ consecutive cells which is significantly smaller than the mini-block size~$\Theta({\varepsilon}^{-3/2})$.
On the other hand, the elements that we did not attempt to insert into a mini-block are only a small fraction of all elements, so the expected locality will be~$\leq \Theta(\varepsilon^{-3/2}+\sqrt{\varepsilon}\cdot \varepsilon^{-2})$.
Intuitively, this means we are replacing the discrepancy of locality over time with a probabilistic discrepancy for every insertion: with some large probability, corresponding to the ``first'' chunk of elements before the load is high, we try to insert near the home location, and with a small probability, corresponding to the ``latest'' chunk of elements, we intentionally insert the element far enough to leave enough local space for the first type of insertions.
After the two-level warm-up, we generalize the same idea to a hierarchy of blocks of larger and larger sizes.

\subsection{Warm-up: a two-level construction}\label{sec:known-upper-warmup}

Put
\[
        \ell=\left\lceil\log(2/\eps)\right\rceil,\qquad
        B=\left\lceil\frac{C\ell}{\eps^2}\right\rceil,\qquad
        b=\left\lceil\frac{C\ell}{\eps^{3/2}}\right\rceil,\qquad
        p=\frac{b}{B},
\]
where $C$ is a sufficiently large constant. Assume for simplicity
that the table is partitioned into consecutive \emph{superblocks} of exactly $B$
cells, each of which is partitioned into $D=B/b$ consecutive
\emph{miniblocks} of exactly $b$ cells.  

Every key first chooses a uniform home superblock, then a uniform home
miniblock inside it, and finally a home cell $h(x)$ in that miniblock. This is simply a uniform table cell.  For each element~$x$ we also draw an
independent reproducible \emph{flexibility bit} which is one with probability $p$.
Every operation first probes $h(x)$.  The keys in each miniblock are stored in
a prefix of that miniblock.  
Thus its load can be read by a binary search requiring~$\Theta(\log (1/\varepsilon))$ probes.

An insertion begins in the miniblock containing $h(x)$.  If the key is
flexible (that is, its flexibility bit was drawn to be positive) or if the mini-block is already full, then the key is sent to a least-loaded miniblock of that superblock: finding one requires~$\Theta(\varepsilon^{-1/2})$ applications of the binary search in all mini-blocks within the super-block. 
If the key is not flexible and the miniblock is not full, then it simply stays in that miniblock.
The selected miniblock stores
the key in the first empty cell of its occupied prefix.  
If all miniblocks in the superblock are
full, the insertion moves forward to the following superblock and repeats the same process as if the key was a flexible key inserted to that next superblock. 

An typical insertion (non flexible key, miniblock is not full) scans just one
miniblock.  Searches use the same flexibility bit.  A flexible key is searched
for in its entire home superblock.  A nonflexible key is first searched for in
its home miniblock.  If the scanned block (either superblock or miniblock) is not full, failure to find the key
certifies that it is absent.  If it is full, the search expands to the entire
superblock, and a full superblock causes it to continue through the following
full superblocks and the first nonfull one.  This is correct by monotonicity:
without deletions, a currently nonfull block was never full, and hence no key
following this rule could have crossed it.

We begin with a useful allocation lemma, giving tail bounds for overflow in bins constructed by interleaving uniform allocations with minimum load allocations. Essentially, we show it has similar overflow probability to what would have happened had there were only the~$(1-p)$ uniform allocations, as the minimum-load allocations actively work against overflows.

\begin{lemma}[Biased minimum allocation]\label{lem:biased-minimum}
Let $z_1,\ldots,z_D$ be bin loads, initially zero.  At every step, independently
of the past, with probability at least $p$ the next item is assigned to a
least-loaded bin, and otherwise its bin is uniform.  If
$\overline z=D^{-1}\sum_jz_j$, then, at every time and for every $g\ge0$,
\[
        \Prb\left[\max_j(z_j-\overline z)\ge g\right]
        \le 8D\exp(-pg/8).
\]
\end{lemma}

\begin{proof}
Assume $p>0$.  Set
$\alpha=p/8\in (0,1/8]$ and define the potential
\[
        \Phi=\sum_{j=1}^D
        \exp\bigl(\alpha(z_j-\overline z)\bigr).
\]
By Jensen's inequality, $\Phi\ge D$.  If the next item is assigned to bin
$J$, then $z_J-\overline z$ increases by $1-1/D$, while every other
$z_j-\overline z$ decreases by $1/D$.  Consequently,
\[
        \Phi'
        =
        e^{-\alpha/D}
        \left(
        \Phi+(e^\alpha-1)
        e^{\alpha(z_J-\overline z)}
        \right).
\]

First suppose that $J$ is uniform.  Averaging over $J$ gives
\[
\begin{aligned}
        \E[\Phi'\mid z_1,\ldots,z_D]
        &=
        e^{-\alpha/D}
        \left(1+\frac{e^\alpha-1}{D}\right)\Phi.
\end{aligned}
\]
Moreover,
\[
\begin{aligned}
 \log\left(
        e^{-\alpha/D}
        \left(1+\frac{e^\alpha-1}{D}\right)
      \right)
 &\le
 -\frac{\alpha}{D}+\frac{e^\alpha-1}{D}
 \le \frac{\alpha^2}{D},
\end{aligned}
\]
where we used $\log(1+x)\le x$ and
$e^\alpha-1-\alpha\le\alpha^2$, valid for $\alpha\le1/8$.
It follows that
\[
        \E[\Phi'\mid z_1,\ldots,z_D]
        \le
        e^{\alpha^2/D}\Phi
        \le
        \left(1+\frac{2\alpha^2}{D}\right)\Phi.
        \tag{1}
\]

Now suppose that $J$ is a least-loaded bin.  Then
$z_J-\overline z\le0$, and hence
$e^{\alpha(z_J-\overline z)}\le1$.  Therefore,
\[
\begin{aligned}
        \Phi'
        &\le
        e^{-\alpha/D}\Phi
        +e^{-\alpha/D}(e^\alpha-1)\\
        &\le
        \left(1-\frac{\alpha}{2D}\right)\Phi+2\alpha,
\end{aligned}
        \tag{2}
\]
using $e^{-x}\le1-x/2$ for $0\le x\le1$ and
$e^\alpha-1\le2\alpha$ for $\alpha\le1/8$.

A minimum choice produces no larger potential than the expected potential
of a uniform choice: indeed, for a least-loaded bin $J$,
\[
        e^{\alpha(z_J-\overline z)}
        \le 1
        \le \frac{\Phi}{D}.
\]
Thus, if the probability of a minimum choice is larger than $p$, we may
upper-bound the expected potential by replacing the excess minimum choices
with uniform choices.  Combining \emph{(1)} and \emph{(2)}, we obtain
\[
\begin{aligned}
        \E[\Phi'\mid z_1,\ldots,z_D]
        &\le
        (1-p)\left(1+\frac{2\alpha^2}{D}\right)\Phi
        +p\left(
        \left(1-\frac{\alpha}{2D}\right)\Phi+2\alpha
        \right)\\
        &=
        \left(
        1+\frac{2(1-p)\alpha^2-p\alpha/2}{D}
        \right)\Phi+2p\alpha.
\end{aligned}
\]
Since $\alpha=p/8$,
\[
        2(1-p)\alpha^2-\frac{p\alpha}{2}
        =-\frac{(1+p)p^2}{32}
        \le-\frac{p^2}{32},
        \qquad
        2p\alpha=\frac{p^2}{4}.
\]
Hence
\[
        \E[\Phi'\mid z_1,\ldots,z_D]
        \le
        \left(1-\frac{p^2}{32D}\right)\Phi+\frac{p^2}{4}.
\]
Since $\Phi=D\leq 8D$ initially, induction gives that we always still have
\[
        \E[\Phi]\le \left(1-\frac{p^2}{32D}\right)8D+\frac{p^2}{4} = 8D.
\]
Finally, if
$\max_j(z_j-\overline z)\ge g$, then $\Phi\ge e^{\alpha g}$.
Markov's inequality therefore gives
\[
\begin{aligned}
        \Prb\left[\max_j(z_j-\overline z)\ge g\right]
        &\le e^{-\alpha g}\E[\Phi]\\
        &\le 8D e^{-pg/8}.
\end{aligned}
\]
\end{proof}

The above lemma is sufficient to claim that if a super-block is not too full, then all of its mini-blocks are not full and in particular typical insertions will not need to explore out of their home mini-block. 
It is thus left to analyze the higher-level behavior: we rephrase a standard linear-probing estimate to control
the top level of how overflow between full super-blocks accumulates. 

\begin{lemma}[Top-block overload]\label{lem:top-block-overload}
At load at most $1-\eps$, send each key from a uniform home superblock to
the first nonfull superblock at or after it.  For a uniform superblock $U$,
let $X$ be the set of elements sent into it by the above process, and let
$K$ be the number of consecutive full blocks beginning at $U$, excluding
$U$ itself.  Thus, $K=0$ if $U$ is not full, and otherwise $K$ is the
number of full blocks immediately following $U$ until the first nonfull
block.  Then
\[
\begin{split}
 \Prb\left[|X|>(1-\eps/2)B\right]
        &\le C_1e^{-c_1\eps^2B},\\
 \E[K]
        &\le C_1e^{-c_1\eps^2B}.
\end{split}
\]
The same bounds hold when $U$ is the home block of any one fixed
previously inserted key.
\end{lemma}

\begin{proof}
Let $M=n/B$ be the number of superblocks, and for every consecutive
interval $I$ of superblocks let $A(I)$ denote the number of keys whose
home superblock lies in $I$.  If $I$ consists of $r$ superblocks, then
\[
        A(I)\sim\text{Bin}\left(N,\frac{r}{M}\right),
\]
where $N\le(1-\eps)MB$ is the current number of keys.  
Standard Chernoff bounds therefore give
\begin{align}
 \Prb\left[A(I)>(r-\eps/2)B\right]
        &\le \exp(-c\eps^2rB),                         \tag{3}\\
 \Prb\left[A(I)\ge rB\right]
        &\le \exp(-c\eps^2rB)                          \tag{4}
\end{align}
for an absolute constant $c>0$.  Indeed, the first threshold exceeds the
mean by at least
\[
        (r-\eps/2)B-(1-\eps)rB
        =\eps B(r-1/2)
        \ge \frac{\eps rB}{2},
\]
while the second exceeds the mean by at least $\eps rB$.

We first bound the load of $U$.  Let $V$ be the first superblock after
the last nonfull superblock preceding $U$.  Such a block exists because
the table is not full.  Every block from $V$ up to but not including $U$
is full.  Moreover, no key whose home block precedes $V$ could have
crossed the nonfull block preceding $V$: since there are no deletions,
that block was never full.

Consequently, if $I$ is the interval from $V$ through $U$, and $I$ has
$r$ blocks, then every key currently stored in $I$ has its home block in
$I$.  Hence, if $|X|>(1-\eps/2)B$, then
\[
        A(I)>(r-1)B+(1-\eps/2)B=(r-\eps/2)B.
\]
For each $r$, there is only one interval of length $r$ ending at $U$.
Taking a union bound over its possible lengths and using~\emph{(3)}, we
obtain
\[
\begin{aligned}
 \Prb\left[|X|>(1-\eps/2)B\right]
 &\le \sum_{r\ge1}\exp(-c\eps^2rB)\\
 &\le C_1\exp(-c_1\eps^2B).
\end{aligned}
\]
As we set $B=\Theta(\ell/\eps^2)$, we have
that $\eps^2B$ is bounded below by a sufficiently large constant.

We next bound $K$.  Fix $k\ge1$ and suppose that $K\ge k$.  Then $U$ and
the first $k$ blocks following $U$ are all full.  Let $V$ again be the
first block after the last nonfull block preceding $U$.  Every block from
$V$ through the $k$-th block following $U$ is therefore full.  If this
interval has $r$ blocks, then $r\ge k+1$, and the same no-crossing
argument gives
\[
        A(I)\ge rB.
\]
For each $r\ge k+1$, there is only one such interval: the interval of
length $r$ ending at the $k$-th block following $U$.  Therefore,
by~\emph{(4)},
\[
\begin{aligned}
        \Prb[K\ge k]
        &\le \sum_{r\ge k+1}\exp(-c\eps^2rB)\\
        &\le C_1\exp(-c_1\eps^2B(k+1)).
\end{aligned}
\]
Using $\E[K]=\sum_{k\ge1}\Prb[K\ge k]$, we conclude that
\[
\begin{aligned}
        \E[K]
        &\le C_1\sum_{k\ge1}
        \exp(-c_1\eps^2B(k+1))\\
        &\le C_1\exp(-c_1\eps^2B),
\end{aligned}
\]
after adjusting the absolute constants.

Finally, suppose that $U$ is the home block of one fixed previously
inserted key $x$.  Conditional on $U$, the home blocks of all other keys
remain independent and uniform.  Thus, for every interval considered
above, its number of home keys is one plus a binomial random variable of
the same form.  The extra one changes each overload threshold by at most
one.  Since our choice of $B$ satisfies $\eps B\ge4$, the gaps used above
remain at least $\eps rB/4$ and the proof still holds.
\end{proof}

We are now ready to analyze the entire algorithm.
\begin{theorem}[Two-level known-load upper bound]
\label{thm:two-level-known-upper}
The two-level construction supports every insertion and every successful or
unsuccessful search with
\[
        \E[T],\ \E[R]
        =
        O\left(\frac{\log(1/\eps)}{\eps^{3/2}}\right).
\]
\end{theorem}

\begin{proof}
Call a superblock good when its load is at most $(1-\eps/2)B$.  Conditional
on its total load, the successive choices among its miniblocks satisfy
Lemma~\ref{lem:biased-minimum}: a fresh flexibility bit gives a minimum choice
with probability $p$, while a nonflexible key has a uniform home miniblock,
and redirecting a full home miniblock only adds minimum choices.  If the
superblock is good, a full miniblock exceeds the average miniblock load by at
least $\eps b/2$. The probability of a non-flexible key going out of its miniblock conditioned on its superblock being good is thus at most
\[
        q
        :=
        \Prb[\text{some miniblock is full}\mid
              \text{the superblock is good}]
        \le
        8D\exp(-p\eps b/16).
\]
The maximum in Lemma~\ref{lem:biased-minimum} is important here: the bound
holds simultaneously for all miniblocks, including the home miniblock of a
key that is searched for later.

With our parameters,
\[
        p\eps b=\Theta(\ell),
        \qquad
        \eps^2B=\Theta(\ell).
\]
By picking a large enough $C$, both $q$ and the upper bounds in
Lemma~\ref{lem:top-block-overload} are at most $\eps^{10}\ll p$.  Outside these
exceptional events, a nonflexible operation scans one miniblock and a flexible
one scans one superblock.  Full superblocks contribute their entire run, whose
expected length is controlled by Lemma~\ref{lem:top-block-overload}.  Therefore,
for either an insertion or a search,
\[
\begin{split}
        \E[T],\ \E[R]
        &=
        O\left(b+pB+qB
          +Be^{-c\eps^2B}\right)\\
        &=O(b)
         =O\left(\frac{\ell}{\eps^{3/2}}\right).
\end{split}
\]
The same argument applies to a successful search for a fixed earlier key.
Lemma~\ref{lem:top-block-overload} includes the case in which the home
superblock belongs to a fixed stored key.  Internally, we bound the event that
\emph{any} miniblock is full, so no conditioning on the searched key's
flexibility bit or home miniblock is required.
\end{proof}

\subsection{The multilevel construction}\label{sec:known-upper-hierarchy}

We now repeat the same balancing idea over several levels.  Fix an integer
\[
        1\le k\le\frac{\ell}{4},
\]
and put
\[
        a=\left(\frac{1}{k\eps}\right)^{1/(k+1)},\qquad
        B_0=\frac{Cka\ell}{\eps},\qquad
        B_i=B_0a^i,\qquad
        p_i=\frac{B_0}{B_i}
        \quad(1\le i\le k).
\]
As before, we assume for simplicity that all block sizes divide correctly.
The table is partitioned into consecutive level-$k$ blocks of size $B_k$.
Each level-$i$ block is partitioned into $a$ consecutive level-$(i-1)$
blocks, down to level-zero blocks of size $B_0$.  Notice that
\begin{equation}\label{eq:hierarchy-parameters}
        B_k=\frac{C\ell}{\eps^2},
        \qquad
        p_iB_i=B_0.
\end{equation}

Every key chooses a uniform home level-$k$ block, then a uniform child at
each successive level, and finally a home cell $h(x)$ in its home level-zero
block.  This is simply a uniform table cell.  For every $1\le i\le k$, the
key also has an independent reproducible flexibility bit $F_i(x)$ which is
one with probability $p_i$.  Every operation first probes $h(x)$.

The insertion procedure is recursive.  Suppose that the insertion is
currently inside a level-$i$ block.  If it was already diverted from its
home path at a higher level, it is sent to a least-loaded child.  Otherwise,
if $F_i(x)=1$, it is also sent to a least-loaded child and becomes diverted.
If neither event occurs, the insertion continues recursively into its
uniform home child.  If that child is full but the current block is not full,
the insertion instead chooses a least-loaded nonfull child and becomes
diverted.  If every child is full, the current block reports that it is full
to its parent.  At level zero, the key is stored in the first empty cell of
the block's occupied prefix.  If its home level-$k$
block is full, the insertion moves forward to the first nonfull level-$k$
block and is treated there as a diverted insertion.

This procedure requires no counters.  The load of a level-zero block can be
found by binary search on its occupied prefix.  Whenever a least-loaded child
of a higher-level block is needed, the algorithm scans that block and counts
the occupied cells in each child.  Once a key becomes diverted, all subsequent
scans occur inside the largest block already scanned.

Searches follow the same hierarchy.  Let $j$ be the largest level for which
$F_j(x)=1$, or let $j=0$ if all the flexibility bits are zero.  The search
first scans the level-$j$ home block containing $h(x)$.  If the key is not
found and this block is nonfull, the search stops.  If it is full, the search
expands to its parent and continues in the same way.  A full home level-$k$
block causes the search to continue through the following full level-$k$
blocks and the first nonfull one.

This search is correct by the same monotonicity argument as before.  A
proactive diversion at level $j$ never leaves the level-$j$ block.  If an
insertion left a smaller home block because of an overflow at a higher level,
then that entire higher-level child was full and remains full.  The search
therefore continues expanding until it reaches a block containing the key.
Conversely, a currently nonfull scanned block was never full, and therefore
certifies that no key could have crossed it.

We leave a different amount of slack at each level.  Put
\[
        \tau_i=\frac{(i+1)\eps}{2(k+1)}
        \qquad(0\le i\le k),
\]
and call a level-$i$ block \emph{good} if its load is at most
$1-\tau_i$.  In particular, a good level-$k$ block has load at most
$1-\eps/2$.  Moreover,
\[
        \tau_i-\tau_{i-1}
        =\frac{\eps}{2(k+1)}.
\]
Thus, inside a good level-$i$ block, a child that is not good must exceed
the average child load by a noticeable amount.

\begin{lemma}[Hierarchical load control]
\label{lem:hierarchical-load-control}
Condition on a good level-$i$ block, where $1\le i\le k$.  The probability
that at least one of its children is not good is at most
\[
        \eta_i
        :=
        8\frac{B_i}{B_{i-1}}
        \exp\left(
          -\frac{p_i\eps B_{i-1}}{16(k+1)}
        \right).
\]
For a sufficiently large constant $C$,
\[
        \max_{1\le i\le k}\eta_i\le\eps^{10}.
\]
Consequently, conditioned on its home level-$k$ block being good, the path
followed by any fixed insertion or search contains only good blocks except
with probability at most $k\eps^{10}$.
\end{lemma}

\begin{proof}
Fix a good level-$i$ block and condition further on the number of keys sent
into it.  Their allocations among its
\[
        D_i=\frac{B_i}{B_{i-1}}=a
\]
children satisfy Lemma~\ref{lem:biased-minimum}.  Indeed, a key which was
already diverted is assigned to a least-loaded child.  Otherwise, its
flexibility bit $F_i(x)$ produces a minimum choice with probability $p_i$,
and in the remaining case its home child is uniform.  Redirecting an
insertion whose home child is full only produces an additional minimum
choice.

Since the parent is good, its average child load is at most
\[
        (1-\tau_i)B_{i-1}.
\]
On the other hand, a child which is not good has load greater than
\[
        (1-\tau_{i-1})B_{i-1}.
\]
Such a child therefore exceeds the average by at least
\[
        (\tau_i-\tau_{i-1})B_{i-1}
        =\frac{\eps B_{i-1}}{2(k+1)}.
\]
Lemma~\ref{lem:biased-minimum} now gives
\[
\begin{aligned}
 \Prb[\text{some child is not good}]
 &\le
 8D_i\exp\left(
   -\frac{p_i\eps B_{i-1}}{16(k+1)}
 \right)\\
 &=\eta_i.
\end{aligned}
\]

By the definition of the parameters,
\[
        \frac{p_i\eps B_{i-1}}{k+1}
        =
        \frac{B_0\eps}{a(k+1)}
        =
        \frac{Ck\ell}{k+1}
        =\Theta(C\ell).
\]
Also, $D_i=a\le\eps^{-1/2}$.  Taking $C$ sufficiently large therefore
ensures that $\eta_i\le\eps^{10}$ simultaneously for every level.

The estimate controls the maximum load among all children, and hence applies
regardless of which child is subsequently selected by the operation.
Applying it successively down the at most $k$ levels and taking a union bound
proves the final assertion.
\end{proof}

\begin{theorem}[$k$-level known-load upper bound]
\label{thm:k-level-known-upper}
For every $1\le k\le\ell/4$, the construction above supports every insertion
and every successful or unsuccessful search with
\[
        \E[T],\ \E[R]
        =
        O\left(
          \frac{k^2\ell}{\eps}
          \left(\frac{1}{k\eps}\right)^{1/(k+1)}
        \right).
\]
\end{theorem}

\begin{proof}
First suppose that the home level-$k$ block is good and that every block
along the path followed by the operation is good.  There are then no
reactive overflows along this path.  If every flexibility bit is zero, the
operation scans only a level-zero block and has probe count and locality
$O(B_0)$.  If $i$ is the largest level for which $F_i(x)=1$, all scans made
by the operation are contained in a level-$i$ block.  Since
$B_{i-1}+B_{i-2}+\cdots+B_0=O(B_i)$, its probe count and locality are both
$O(B_i)$.

The probability that $i$ is the largest positive flexibility level is at
most $p_i$.  The expected cost outside the exceptional events is therefore
\[
\begin{aligned}
        O\left(B_0+\sum_{i=1}^kp_iB_i\right)
        &=
        O\left(B_0+\sum_{i=1}^kB_0\right)\\
        &=O(kB_0).
\end{aligned}
        \tag{5}
\]

It remains to bound the exceptional events.  Since
\[
        \tau_k=\frac{\eps}{2},
\]
a home level-$k$ block is not good precisely when its load exceeds
$(1-\eps/2)B_k$.  Lemma~\ref{lem:top-block-overload} and
\[
        \eps^2B_k=C\ell
\]
show, after increasing $C$, that this happens with probability at most
$\eps^{10}$.  The same lemma shows that the expected number $K$ of additional
full top blocks traversed by the operation is at most $\eps^{10}$.

Conditioned on a good home top block,
Lemma~\ref{lem:hierarchical-load-control} shows that an internal failure has
probability at most $k\eps^{10}$.  We may pessimistically charge $O(B_k)$
for every internal failure or nonfull exceptional top block.  If the home
top block is full, the operation scans at most $K+2$ top blocks, giving
probe count and locality $O((K+2)B_k)$.  Thus, the total expected
contribution of all exceptional events is
\[
        O\left((k+1)\eps^{10}B_k\right)
        =o(B_0).
        \tag{6}
\]

Combining \emph{(5)} and \emph{(6)}, and recalling the definition of $B_0$,
gives
\[
\begin{aligned}
        \E[T],\ \E[R]
        &=O(kB_0)\\
        &=
        O\left(
          \frac{k^2a\ell}{\eps}
        \right)\\
        &=
        O\left(
          \frac{k^2\ell}{\eps}
          \left(\frac{1}{k\eps}\right)^{1/(k+1)}
        \right).
\end{aligned}
\]

For an insertion or an unsuccessful search, the home top block is uniform
and independent of the current table.  For a successful search for any fixed
previously inserted key, the tagged version of
Lemma~\ref{lem:top-block-overload} gives the same top-level bounds.  The
internal load estimate controls the maximum child load at each level and
therefore also controls the home path of the searched key.  Hence the same
bound holds for successful searches.
\end{proof}

\begin{proof}[Proof of Theorem~\ref{thm:known-target-upper}]
For all sufficiently small $\eps$, take
\[
        k=\left\lfloor\frac{\ell}{4}\right\rfloor.
\]
Then
\[
        a
        =
        \left(\frac{1}{k\eps}\right)^{1/(k+1)}
        =O(1),
\]
and Theorem~\ref{thm:k-level-known-upper} gives
\[
        \E[T],\ \E[R]
        =
        O\left(\frac{\ell^3}{\eps}\right)
        =
        O\left(\frac{\log^3(1/\eps)}{\eps}\right).
\]

For the remaining constant range of $\eps$, linear probing has constant
expected probe count and locality, and hence also satisfies the claimed bound.

Every leaf block is filled monotonically from left to right, all other load
information is obtained by scanning ordinary table cells, and no key is ever
moved after being inserted.  Thus the construction is an immutable
open-addressing scheme and the theorem follows.
\end{proof}

\section{Discussion and Open Problems}\label{sec:open-problems}

We initiated the study of \emph{locality} as an alternative measure to probe
count in open-addressed hash tables.  In the standard setting, in which one
algorithm must work at all loads, our lower bound shows that the quadratic
locality profile of linear probing is optimal.  When the final load
$1-\eps$ is known in advance, our amortized lower bound gives
$\Omega(1/\eps)$, while our upper bound gives
$\widetilde O(1/\eps)$ expected probe count and locality for every individual
insertion and search.

The first remaining question is therefore whether the logarithmic factors in
the known-load upper bound can be removed.

\begin{question}
Suppose the algorithm is given a target slack $\eps$ in advance and is
promised that the load will never exceed $1-\eps$.  Can every insertion and
every successful or unsuccessful search be performed with
\[
        \E[T],\ \E[R]=O(1/\eps)?
\]
\end{question}

Such a result would exactly match the amortized locality lower bound while
fully deamortizing it across the insertion sequence.  Conversely, it would be
interesting to prove any separation between the amortized lower bound and the
worst expected cost of an individual operation.

The recent breakthroughs mentioned throughout the introduction introduce
open-addressing schemes, including no-reordering constructions, whose probe
complexities are below the classical $\Theta(1/\eps)$ greedy bound.  Our lower
bound shows that such improvements cannot improve the all-load locality scale:
locality $1/\eps^2$ is unavoidable.  It remains possible, however, that a
non-greedy placement scheme can retain this optimal locality while using fewer
probes.

\begin{question}
Can a no-reordering table combine optimal all-load locality
$O(1/\eps^2)$ with insertion probe complexity $o(1/\eps)$ by using
non-greedy placement?
\end{question}

Our load-oblivious greedy construction has optimal expected probe count
$O(1/\eps)$ and satisfies $R=O(T^2)$ deterministically, but we do not know how
to bound its expected locality.

\begin{question}
Is there a load-oblivious immutable open-addressing scheme which, at every
load $1-\eps$, satisfies
\[
        \E[T]=O(1/\eps)
        \qquad\text{and}\qquad
        \E[R]=O(1/\eps^2)?
\]
\end{question}

For our particular construction, the missing ingredient is a stronger tail
bound on occupied-cell densities.  The variance lemma gives enough control for
the first moment of the probe count, but not for the expected locality.  This
motivates the following more general question.

\begin{question}
Consider any symmetric probing mechanism after $t$ insertions, and let
$O_t\subseteq\Zn$ be the set of occupied cells.  For every fixed set
$I\subseteq\Zn$, let
\[
        X_I=|I\cap O_t|.
\]
What moment or tail bounds, beyond
$\Var(X_I)\le |I|$, hold uniformly over all symmetric probing mechanisms?
\end{question}

Finally, while locality is particularly natural for immutable tables, it can
also be defined for schemes that move or rebuild keys.  Understanding the
tradeoff between locality, probe count, and the amount of reordering in such
tables remains an interesting direction.

\bibliography{refs}

@book{Knuth1998,
  author    = {Knuth, Donald E.},
  title     = {The Art of Computer Programming, Volume 3: Sorting and Searching},
  edition   = {2},
  publisher = {Addison-Wesley},
  address   = {Reading, MA},
  year      = {1998},
  isbn      = {978-0-201-89685-5}
}

@article{FlajoletPobleteViola1998,
  author  = {Flajolet, Philippe and Poblete, Patricio and Viola, Alfredo},
  title   = {On the Analysis of Linear Probing Hashing},
  journal = {Algorithmica},
  volume  = {22},
  number  = {4},
  pages   = {490--515},
  year    = {1998},
  doi     = {10.1007/PL00009236},
  url     = {https://doi.org/10.1007/PL00009236}
}

@article{Yao1985,
  author  = {Yao, Andrew C.-C.},
  title   = {Uniform Hashing Is Optimal},
  journal = {Journal of the ACM},
  volume  = {32},
  number  = {3},
  pages   = {687--693},
  year    = {1985},
  doi     = {10.1145/3828.3836},
  url     = {https://doi.org/10.1145/3828.3836}
}

@inproceedings{KuszmaulXi2024,
  author    = {Kuszmaul, William and Xi, Zoe},
  title     = {Towards an Analysis of Quadratic Probing},
  booktitle = {51st International Colloquium on Automata, Languages, and Programming (ICALP 2024)},
  series    = {Leibniz International Proceedings in Informatics (LIPIcs)},
  volume    = {297},
  pages     = {103:1--103:19},
  publisher = {Schloss Dagstuhl -- Leibniz-Zentrum f{\"u}r Informatik},
  address   = {Dagstuhl, Germany},
  year      = {2024},
  doi       = {10.4230/LIPIcs.ICALP.2024.103},
  url       = {https://drops.dagstuhl.de/entities/document/10.4230/LIPIcs.ICALP.2024.103},
  isbn      = {978-3-95977-322-5},
  issn      = {1868-8969}
}

@article{Tonks1936,
  author  = {Tonks, Lewi},
  title   = {The Complete Equation of State of One, Two and Three-Dimensional Gases of Hard Elastic Spheres},
  journal = {Physical Review},
  volume  = {50},
  number  = {10},
  pages   = {955--963},
  year    = {1936},
  doi     = {10.1103/PhysRev.50.955},
  url     = {https://doi.org/10.1103/PhysRev.50.955}
}

@article{AlderWainwright1957,
  author  = {Alder, B. J. and Wainwright, T. E.},
  title   = {Phase Transition for a Hard Sphere System},
  journal = {The Journal of Chemical Physics},
  volume  = {27},
  number  = {5},
  pages   = {1208--1209},
  year    = {1957},
  doi     = {10.1063/1.1743957},
  url     = {https://doi.org/10.1063/1.1743957}
}

@book{Forrester2010,
  author    = {Forrester, Peter J.},
  title     = {Log-Gases and Random Matrices},
  series    = {London Mathematical Society Monographs},
  number    = {34},
  publisher = {Princeton University Press},
  address   = {Princeton, NJ},
  year      = {2010},
  isbn      = {978-0-691-12829-0},
  url       = {https://www.jstor.org/stable/j.ctt7t5vq}
}

@misc{Serfaty2017,
  author        = {Serfaty, Sylvia},
  title         = {Microscopic Description of {Log} and {Coulomb} Gases},
  year          = {2017},
  eprint        = {1709.04089},
  archivePrefix = {arXiv},
  primaryClass  = {math-ph},
  doi           = {10.48550/arXiv.1709.04089},
  url           = {https://arxiv.org/abs/1709.04089}
}

@article{MauryRoudneffSantambrogioVenel2011,
  author  = {Maury, Bertrand and Roudneff-Chupin, Aude and Santambrogio, Filippo and Venel, Juliette},
  title   = {Handling Congestion in Crowd Motion Modeling},
  journal = {Networks and Heterogeneous Media},
  volume  = {6},
  number  = {3},
  pages   = {485--519},
  year    = {2011},
  doi     = {10.3934/nhm.2011.6.485},
  url     = {https://doi.org/10.3934/nhm.2011.6.485}
}

@article{LeclercMerigotSantambrogioStra2020,
  author  = {Leclerc, Hugo and M{\'e}rigot, Quentin and Santambrogio, Filippo and Stra, Federico},
  title   = {Lagrangian Discretization of Crowd Motion and Linear Diffusion},
  journal = {SIAM Journal on Numerical Analysis},
  volume  = {58},
  number  = {4},
  pages   = {2093--2118},
  year    = {2020},
  doi     = {10.1137/19M1274201},
  url     = {https://doi.org/10.1137/19M1274201}
}

@incollection{boucheron2004concentration,
  author    = {Boucheron, St{\'e}phane and Lugosi, G{\'a}bor and Bousquet, Olivier},
  title     = {Concentration Inequalities},
  booktitle = {Advanced Lectures on Machine Learning: ML Summer Schools 2003, Canberra, Australia, February 2--14, 2003, T{\"u}bingen, Germany, August 4--16, 2003, Revised Lectures},
  editor    = {Bousquet, Olivier and von Luxburg, Ulrike and R{\"a}tsch, Gunnar},
  series    = {Lecture Notes in Computer Science},
  volume    = {3176},
  pages     = {208--240},
  publisher = {Springer},
  year      = {2004},
  doi       = {10.1007/978-3-540-28650-9_9}
}

@article{AggarwalVitter1988,
  author  = {Aggarwal, Alok and Vitter, Jeffrey Scott},
  title   = {The Input/Output Complexity of Sorting and Related Problems},
  journal = {Communications of the ACM},
  volume  = {31},
  number  = {9},
  pages   = {1116--1127},
  year    = {1988},
  doi     = {10.1145/48529.48535},
  url     = {https://doi.org/10.1145/48529.48535}
}

@article{JensenPagh2008,
  author  = {Jensen, Morten Skaarup and Pagh, Rasmus},
  title   = {Optimality in External Memory Hashing},
  journal = {Algorithmica},
  volume  = {52},
  number  = {3},
  pages   = {403--411},
  year    = {2008},
  doi     = {10.1007/s00453-007-9155-x},
  url     = {https://doi.org/10.1007/s00453-007-9155-x}
}

@inproceedings{IaconoPatrascu2012,
  author    = {Iacono, John and P{\u a}tra{\c s}cu, Mihai},
  title     = {Using Hashing to Solve the Dictionary Problem (in External Memory)},
  booktitle = {Proceedings of the Twenty-Third Annual ACM-SIAM Symposium on Discrete Algorithms (SODA)},
  pages     = {570--582},
  publisher = {SIAM},
  year      = {2012},
  doi       = {10.1137/1.9781611973099.48},
  url       = {https://doi.org/10.1137/1.9781611973099.48}
}

@inproceedings{ConwayFarachColtonShilane2018,
  author    = {Conway, Alex and Farach-Colton, Mart{\'i}n and Shilane, Philip},
  title     = {Optimal Hashing in External Memory},
  booktitle = {45th International Colloquium on Automata, Languages, and Programming (ICALP 2018)},
  series    = {Leibniz International Proceedings in Informatics (LIPIcs)},
  volume    = {107},
  pages     = {39:1--39:14},
  publisher = {Schloss Dagstuhl -- Leibniz-Zentrum f{\"u}r Informatik},
  address   = {Dagstuhl, Germany},
  year      = {2018},
  doi       = {10.4230/LIPIcs.ICALP.2018.39},
  url       = {https://doi.org/10.4230/LIPIcs.ICALP.2018.39}
}

@article{Maurer1968,
  author  = {Maurer, Ward Douglas},
  title   = {Programming Technique: An Improved Hash Code for Scatter Storage},
  journal = {Communications of the ACM},
  volume  = {11},
  number  = {1},
  pages   = {35--38},
  year    = {1968},
  doi     = {10.1145/362851.362856},
  url     = {https://doi.org/10.1145/362851.362856}
}

@article{HopgoodDavenport1972,
  author  = {Hopgood, F. R. A. and Davenport, J.},
  title   = {The Quadratic Hash Method When the Table Size Is a Power of 2},
  journal = {The Computer Journal},
  volume  = {15},
  number  = {4},
  pages   = {314--315},
  year    = {1972},
  doi     = {10.1093/comjnl/15.4.314},
  url     = {https://doi.org/10.1093/comjnl/15.4.314}
}

@article{Ecker1974,
  author  = {Ecker, A.},
  title   = {The Period of Search for the Quadratic and Related Hash Methods},
  journal = {The Computer Journal},
  volume  = {17},
  number  = {4},
  pages   = {340--343},
  year    = {1974},
  doi     = {10.1093/comjnl/17.4.340},
  url     = {https://doi.org/10.1093/comjnl/17.4.340}
}

@article{Batagelj1975,
  author  = {Batagelj, Vladimir},
  title   = {The Quadratic Hash Method When the Table Size Is Not a Prime Number},
  journal = {Communications of the ACM},
  volume  = {18},
  number  = {4},
  pages   = {216--217},
  year    = {1975},
  doi     = {10.1145/360715.360729},
  url     = {https://doi.org/10.1145/360715.360729}
}

@article{Radke1970,
  author  = {Radke, Charles E.},
  title   = {The Use of Quadratic Residue Research},
  journal = {Communications of the ACM},
  volume  = {13},
  number  = {2},
  pages   = {103--105},
  year    = {1970},
  doi     = {10.1145/361953.361968},
  url     = {https://doi.org/10.1145/361953.361968}
}

@book{CormenLeisersonRivestStein2009,
  author    = {Cormen, Thomas H. and Leiserson, Charles E. and Rivest, Ronald L. and Stein, Clifford},
  title     = {Introduction to Algorithms},
  edition   = {3},
  publisher = {MIT Press},
  address   = {Cambridge, MA},
  year      = {2009},
  isbn      = {978-0-262-03384-8}
}

@book{Weiss2000,
  author    = {Weiss, Mark Allen},
  title     = {Data Structures and Problem Solving Using C++},
  edition   = {2},
  publisher = {Addison-Wesley},
  address   = {Reading, MA},
  year      = {2000},
  isbn      = {978-0-201-61250-9}
}

@inproceedings{PurcellHarris2005,
  author    = {Purcell, Chris and Harris, Tim},
  title     = {Non-blocking Hashtables with Open Addressing},
  booktitle = {Proceedings of the 19th International Symposium on Distributed Computing (DISC)},
  series    = {Lecture Notes in Computer Science},
  volume    = {3724},
  pages     = {108--121},
  publisher = {Springer},
  year      = {2005},
  doi       = {10.1007/11561927_10},
  url       = {https://doi.org/10.1007/11561927_10}
}

@inproceedings{AskitisZobel2005,
  author    = {Askitis, Nikolas and Zobel, Justin},
  title     = {Cache-Conscious Collision Resolution in String Hash Tables},
  booktitle = {String Processing and Information Retrieval},
  series    = {Lecture Notes in Computer Science},
  volume    = {3772},
  pages     = {91--102},
  publisher = {Springer},
  year      = {2005},
  doi       = {10.1007/11575832_11},
  url       = {https://doi.org/10.1007/11575832_11}
}

@inproceedings{BenderKuszmaulZhou2024,
  author    = {Bender, Michael A. and Kuszmaul, William and Zhou, Renfei},
  title     = {Tight Bounds for Classical Open Addressing},
  booktitle = {Proceedings of the 2024 IEEE 65th Annual Symposium on Foundations of Computer Science (FOCS)},
  pages     = {636--657},
  publisher = {IEEE},
  year      = {2024},
  doi       = {10.1109/FOCS61266.2024.00047},
  url       = {https://doi.org/10.1109/FOCS61266.2024.00047},
  eprint    = {2409.11280},
  archivePrefix = {arXiv},
  primaryClass  = {cs.DS}
}

@inproceedings{FarachColtonKrapivinKuszmaul2025,
  author    = {Farach-Colton, Mart{\'i}n and Krapivin, Andrew and Kuszmaul, William},
  title     = {Optimal Bounds for Open Addressing Without Reordering},
  booktitle = {Proceedings of the 2024 IEEE 65th Annual Symposium on Foundations of Computer Science (FOCS)},
  pages     = {594--605},
  publisher = {IEEE},
  year      = {2024},
  doi       = {10.1109/FOCS61266.2024.00045},
  url       = {https://doi.org/10.1109/FOCS61266.2024.00045},
  eprint    = {2501.02305},
  archivePrefix = {arXiv},
  primaryClass  = {cs.DS}
}

@inproceedings{FarachColtonKrapivinKuszmaul2026,
  author    = {Farach-Colton, Mart{\'i}n and Krapivin, Andrew and Kuszmaul, William},
  title     = {Greedy Open Addressing Revisited: Beyond Yao's Lower Bound},
  booktitle = {Proceedings of the 58th Annual ACM Symposium on Theory of Computing (STOC)},
  pages     = {1116--1127},
  publisher = {Association for Computing Machinery},
  year      = {2026},
  doi       = {10.1145/3798129.3800823},
  url       = {https://doi.org/10.1145/3798129.3800823}
}

@inproceedings{ChesettiShiPhillipsPandey2025,
  author    = {Chesetti, Yuvaraj and Shi, Benwei and Phillips, Jeff M. and Pandey, Prashant},
  title     = {Zombie Hashing: Reanimating Tombstones in a Graveyard},
  booktitle = {Proceedings of the ACM on Management of Data},
  volume    = {3},
  number    = {3},
  articleno = {236},
  pages     = {1--27},
  year      = {2025},
  doi       = {10.1145/3725424},
  url       = {https://doi.org/10.1145/3725424}
}

@inproceedings{BenderKuszmaulKuszmaul2021,
  author    = {Bender, Michael A. and Kuszmaul, Bradley C. and Kuszmaul, William},
  title     = {Linear Probing Revisited: Tombstones Mark the Demise of Primary Clustering},
  booktitle = {Proceedings of the 2021 IEEE 62nd Annual Symposium on Foundations of Computer Science (FOCS)},
  pages     = {1171--1182},
  publisher = {IEEE},
  year      = {2022},
  doi       = {10.1109/FOCS52979.2021.00115},
  url       = {https://doi.org/10.1109/FOCS52979.2021.00115},
  eprint    = {2107.01250},
  archivePrefix = {arXiv},
  primaryClass  = {cs.DS}
}

@article{DietzfelbingerWeidling2007,
  author  = {Dietzfelbinger, Martin and Weidling, Christoph},
  title   = {Balanced Allocation and Dictionaries with Tightly Packed Constant Size Bins},
  journal = {Theoretical Computer Science},
  volume  = {380},
  pages   = {47--68},
  year    = {2007},
  doi     = {10.1016/j.tcs.2007.02.054},
  url     = {https://doi.org/10.1016/j.tcs.2007.02.054}
}
\bibliographystyle{alpha}

\end{document}